\newif\ifnotes
\notesfalse

\documentclass[11pt]{article}

\usepackage{amsmath}
\usepackage{amsfonts}
\usepackage{amssymb}
\usepackage{amsthm}
\usepackage{bm}
\usepackage{bbm}
\usepackage{enumerate}
\usepackage{algorithm}
\usepackage{algorithmicx}
\usepackage{algpseudocode}
\usepackage{xspace}
\usepackage{color}
\usepackage{fullpage}
\usepackage{url}
\usepackage[dvips,pdfstartview=FitH,pdfpagemode=UseNone,colorlinks=true,citecolor=blue,linkcolor=blue]{hyperref}
\usepackage[capitalise,noabbrev]{cleveref}

\usepackage{tikz}
\usepackage{pgfplots}
\usetikzlibrary{positioning}

\allowdisplaybreaks

\newtheorem{theorem}{Theorem}[section]
\newtheorem{lemma}[theorem]{Lemma}
\newtheorem{definition}[theorem]{Definition}
\newtheorem{corollary}[theorem]{Corollary}
\newtheorem{example}[theorem]{Example}

\newtheorem{claim}[theorem]{Claim}
\newtheorem{conjecture}[theorem]{Conjecture}
\newtheorem{remark}[theorem]{Remark}

\theoremstyle{definition}
\newtheorem{protocol}[theorem]{Protocol}

\newcommand{\F}{\mathbb{F}}

\newcommand{\poly}{\mathsf{poly}}

\newcommand{\agree}{\nice{agree}}

\newcommand{\ball}{B}

\newcommand{\RS}{\mathsf{RS}}

\newcommand{\calL}{\mathcal{L}}

\newcommand{\E}{\mathbf{E}}

\newcommand{\Tr}{\mathsf{Tr}}
\newcommand{\AVG}{\mathsf{AVG}}
\newcommand{\ZERO}{\mathsf{ZERO}}

\newcommand{\nice}[1]{{\sf{#1}}\xspace}
\newcommand{\set}[1]{\left\{#1\right\}}
\newcommand{\condset}[2]{\set{#1 \mid #2}}

\newcommand{\DFRI}{\nice{DEEP-FRI}}
\newcommand{\DEEPSTIK}{\nice{DEEP-ALI}}
\newcommand{\ALI}{\nice{ALI}}
\newcommand{\FRI}{\nice{FRI}}
\newcommand{\N}{{\mathbb{N}}}
\newcommand{\RateInt}{{\cal{R}}}

\newcommand{\zr}{^{(0)}}

\newcommand{\one}{^{(1)}}
\newcommand{\two}{^{(2)}}
\newcommand{\ii}{^{(i)}}

\newcommand{\iim}{^{(i-1)}}
\newcommand{\iip}{^{(i+1)}}

\newcommand{\fin}{^{(\RecInt)}}

\newcommand{\Rate}{\rho}
\newcommand{\eqndef}{\triangleq}
\newcommand{\rounds}{\nice{r}}
\newcommand{\RecInt}{\rounds}
\newcommand{\commit}{\nice{COMMIT}}
\newcommand{\query}{\nice{QUERY}}

\newcommand{\DistanceLi}[1]{\Delta\ii\left(#1\right)}

\newcommand{\acc}{\nice{accept}}

\newcommand{\RepInt}{\ell}

\newcommand{\coset}{S}
\newcommand{\cosets}{{\cal{S}}}
\newcommand{\error}{\nice{err}}

\newcommand{\QUOTIENT}{\mathsf{QUOTIENT}}
\newcommand{\CQ}{\mathsf{CorrectedQUOTIENT}}
\newcommand{\substar}{{_{\star}}}

\newcommand{\NTIME}{\nice{NTIME}}
\DeclareMathOperator{\RPT}{RPT}

\newcommand{\eli}[1]{\ifnotes{\color{cyan} Eli: #1}\xspace\fi}
\newcommand{\swastik}[1]{\ifnotes{\color{red} Swastik: #1}\xspace\fi}
\newcommand{\shubhangi}[1]{\ifnotes{\color{magenta} Shubhangi: #1}\xspace\fi}
\newcommand{\lior}[1]{\ifnotes{\color{blue} Lior: #1}\xspace\fi}

\newcommand{\Ls}{\nice{List}}
\newcommand{\bu}{\bar{u}}
\newcommand{\bv}{\bar{v}}

\newcommand{\event}{{\cal{E}}}
\newcommand{\slack}{{\delta}}
\newcommand{\ohalf}{^{(1.5)}}
\newcommand{\deltastar}{\delta^*}
\newcommand{\Lstar}{\calL^*}
\newcommand{\nustar}{{\nu^*}}
\newcommand{\johnsoneps}{\varepsilon}

\begin{document}
	\title{DEEP-FRI: Sampling Outside the Box Improves Soundness}
\author{Eli Ben-Sasson\thanks{StarkWare Industries Ltd. {\tt \{eli,lior\}@starkware.co}} \and Lior Goldberg\footnotemark[1] \and
	Swastik Kopparty\thanks{Department of Mathematics and Department of Computer Science, Rutgers University. Research supported in part by NSF grants CCF-1253886, CCF-1540634, CCF-1814409 and CCF-1412958, and BSF grant 2014359. Some of this research was done while visiting the Institute for Advanced Study. {\tt swastik.kopparty@gmail.com}}
\and Shubhangi Saraf\thanks{Department of Mathematics and Department of Computer Science, Rutgers University. Research supported in part by NSF grants
CCF-1350572, CCF-1540634 and CCF-1412958, BSF grant 2014359, a Sloan research fellowship
and the Simons Collaboration on Algorithms and Geometry. Some of this research was done while visiting the Institute for Advanced Study.  {\tt shubhangi.saraf@gmail.com}}}
\date{}
\maketitle

\eli{test} \swastik{test} \shubhangi{test} \lior{test}

\begin{abstract}
	Motivated by the quest for scalable and succinct zero knowledge arguments,
	we revisit worst-case-to-average-case reductions for
	linear spaces, raised by [Rothblum, Vadhan, Wigderson, STOC 2013]. The previous state of the art by [Ben-Sasson, Kopparty, Saraf, CCC 2018] showed that if some member of an affine space $U$ is $\delta$-far in relative Hamming distance from a linear code $V$ --- this is the worst-case assumption --- then most elements of $U$ are almost-$\delta$-far from $V$ --- this is the average case. However, this result was known to hold only below the ``double Johnson'' function of the relative distance $\delta_V$ of the code $V$, i.e., only when $\delta<1-(1-\delta_V)^{1/4}$.

	First, we increase the soundness-bound to the
	``one-and-a-half Johnson'' function of $\delta_V$ and show that the  average distance of $U$ from $V$ is nearly $\delta$ for any worst-case distance $\delta$ smaller than $1-(1-\delta_V)^{1/3}$.
	This bound is tight, which is somewhat surprising because the one-and-a-half Johnson function is unfamiliar in the literature on error correcting codes.

	To improve soundness further for Reed Solomon codes we sample  outside the box. We suggest a new protocol in which the
	verifier samples a single point $z$ outside the box $D$ on which codewords are evaluated, and asks the prover for the value at $z$ of the interpolating polynomial of a random element of $U$. Intuitively, the answer provided by the prover ``forces'' it to choose one codeword from a list of ``pretenders'' that are close to $U$.
	We call this technique Domain Extending for Eliminating Pretenders (DEEP).

	The DEEP method improves the soundness of the worst-case-to-average-case reduction for RS codes up their list decoding radius. This radius is bounded from below by the Johnson bound, implying average distance is approximately $\delta$ for all $\delta<1-(1-\delta_V)^{1/2}$. Under a plausible conjecture about the list decoding radius of Reed-Solomon codes, average distance from $V$ is approximately $\delta$ for all $\delta$. The DEEP technique can be generalized to all linear codes, giving improved reductions for capacity-achieving list-decodable codes.

	Finally, we use the DEEP technique to devise two new protocols:
	\begin{itemize}
		\item An Interactive Oracle Proof of Proximity (IOPP) for RS codes, called \DFRI.
		This soundness of the protocol improves upon that of the \FRI protocol of [Ben-Sasson et al., ICALP 2018] while retaining linear arithmetic proving complexity and logarithmic verifier arithmetic complexity.
		\item  An Interactive Oracle Proof (IOP) for the Algebraic Linking IOP (ALI) protocol used
		to construct zero knowledge scalable transparent arguments of knowledge (ZK-STARKs) in [Ben-Sasson et al., eprint 2018].
		The new protocol, called \DEEPSTIK, improves soundness of this crucial step from a small constant $ < 1/8$ to a constant arbitrarily close to $1$.
	\end{itemize}
	\thispagestyle{empty}
\end{abstract}

\pagenumbering{arabic}

\section{Introduction}\label{sec:intro}

Arithmetization is a marvelous technique that can be used to reduce problems in computational complexity, like verifying membership in a nondeterministic language, to questions about membership of vectors in algebraic codes like Reed-Solomon (RS) and Reed-Muller (RM) codes~\cite{Razborov87,LundFKN92}.
One of the end-points of such a reduction is the  RS proximity testing (RPT) problem. It is a problem of inherent theoretical interest, but also of significant
practical importance because it is used in recent
constructions of succinct zero knowledge (ZK) arguments including Ligero~\cite{ligero}, Aurora~\cite{Aurora}, and  Scalable Transparent ARguments of Knowledge (ZK-STARKs)~\cite{stark}. We discuss this connection after
describing the problem and our results.

In the RPT problem
a verifier is given oracle access to a function $f:D\to \F$, we call $D\subset \F$ the \emph{evaluation domain}, and
is tasked with distinguishing between the ``good'' case that $f$ is a polynomial of degree at most $d$
and the ``bad'' case in which $f$ is $\delta$-far in relative
Hamming distance from all degree-$d$ polynomials. To achieve succinct verification time, poly-logarithmic in $d$,
we must allow the verifier some form of interaction with a prover --- the party claiming that $\deg(f)\leq d$. Initially, this interaction took the form of oracle access to a probabilistically checkable proof of proximity (PCPP) \cite{BenSassonGHSV06} provided by the prover in addition to $f$. Indeed, in this model the RPT problem can be ``solved'' with PCPPs of quasilinear size $|D|\poly\log |D|$, constant
query complexity and constant soundness~\cite{BS08,Din07}. However, the concrete complexity of prover time, verifier time and communication complexity are rather large, even when considering practical settings that involve moderately small instance sizes.

To improve prover, verifier, and communication complexity for concrete (non-asymptotic) size problems, the Interactive Oracle Proofs of Proximity (IOPP) model is more suitable~\cite{ReingoldRR16,BenSassonCS16,Ben-SassonCFGRS16}.
This model can be viewed as a multi-round PCPP. Instead of having the prover write down a single proof $\pi$, in the IOPP setting the proof oracle is produced over a number of rounds of interaction, during which the verifier sends random bits and the prover responds with additional (long) messages to which the verifier is allowed oracle access. The additional
rounds of interaction allow for a dramatic improvement in the asymptotic and concrete complexity of solving the RPT problem. In particular, the Fast RS IOPP (\FRI) protocol of \cite{FRI} has linear prover arithmetic complexity, logarithmic verifier arithmetic complexity and constant soundness. Our goal here is to improve soundness of this protocol and to suggest better protocols in terms of soundness in the high-error regime
(also known as the ``list decoding'' regime).


Soundness analysis of \FRI reduces to the following natural ``worst-case-to-average-case'' question regarding linear spaces, which is also independently very interesting for the case of general (non-RS) codes.  This question was originally raised in a different setting by \cite{RVW2013} and we start by discussing it for general linear codes before focusing on the special, RS code, case.

\subsection{Maximum distance vs average distance to a linear code}
Suppose that $U\subset \F^D$ is a ``line'', a $1$-dimensional\footnote{The generalization of our results to spaces $U$ of dimension $>1$ is straightforward by partitioning $U$ into lines through $u^*$ and applying these results to each line.} affine space over $\F$.
Let $u^*\in \F^D$ denote the origin of this line
and $u$ be its slope, so that $U=\condset{u_x=u^*+x u}{x\in \F}$. For a fixed linear space $V\subset \F^D$, pick $u^*$ to be the element in $U$ that is farthest from $V$, denoting by $\delta_{\max}$ its relative Hamming distance (from $V$). This is our worst-case assumption. Letting
$\delta_x=\Delta(u_x,V)$ where $\Delta$ denotes relative Hamming distance, what can be said about the \emph{expected} distance $\E_{x\in \F} [\delta_x]$ of $u_x$ from $V$?

Rothblum, Vadhan and Wigderson showed that $\E_x[\delta_x]\geq \frac{\delta_{\max}}{2}-o(1)$ for all spaces $U$ and $V$, where, here and below, $o(1)$ denotes negligible
terms that approach $0$ as $|\F|\to\infty$ \cite{RVW2013}. A subset of the co-authors of this paper
improved this to $\E[\delta_x]\geq 1-\sqrt{1-\delta_{\max}}-o(1)$, showing the average distance scales roughly like the Johnson list-decoding function of $\delta_{\max}$, where $J(x):=1-\sqrt{1-x}$~\cite{Ben-SassonKS18}. In both of these bounds the expected distance is strictly smaller than $\delta_{\max}$. However,
the latter paper also showed that when $V$ is a (linear) error correcting code with large relative distance $\delta_V$, if $\delta_{\max}$ is smaller than the ``double Johnson'' function of $\delta_V$,
given by $J\two(x):=J(J(x))$, then the average distance hardly deteriorates,
\begin{equation}
\label{eq:double johnson}
\E[\delta_x]\geq \min\left(\delta_{\max}, J\two(\delta_V)\right)-o(1) = \min\left(\delta_{\max}, 1-\sqrt[4]{1-\delta_V}\right)-o(1)
\end{equation}
and the equation above summarizes the previous state of affairs on this matter.

Our first result is an improvement of \cref{eq:double johnson} to
the ``one-and-a-half-Johnson'' function $J\ohalf(x)=1-(1-x)^{1/3}$.
\cref{lem:one and half Johnson} says that for codes $V$ of relative Hamming distance $\delta_V$,

\begin{equation}
\label{eq:one and a half johnson}
\E[\delta_x]\geq \min\left(\delta_{\max}, J\ohalf(\delta_V)\right)-o(1)
 = \min\left(\delta_{\max}, 1-\sqrt[3]{1-\delta_V}\right)-o(1).
\end{equation}

Our second result shows that \cref{eq:one and a half johnson} is tight,
even for the special case of $V$ being an RS code. We find
this result somewhat surprising because the $J\ohalf(x)$ function
is not known to be related to any meaningful coding theoretic notion.
The counter-example showing the tightness of \cref{eq:one and a half johnson} arises for very special cases, in which (i) $\F$ is a binary field (of characteristic $2$), (ii) the rate $\rho$ is precisely $1/8=2^{-3}$ and, most importantly, (iii) the evaluation domain $D$ equals all of $\F$ (see \cref{sec:tightness discussion}).
Roughly speaking, the counter-example uses functions $u^*,u:\F_{2^n}\to\F_{2^n}$ that are $3/4=1-\rho^{2/3}$-far from polynomials of degree $\rho 2^n$ yet \emph{pretend} to be low-degree because for all $x\in \F_{2^n}\setminus\set{0}$ the function $u^*+xu$ is  $1/2=\sqrt[3]{\rho}$-close to a polynomial of degree $\rho2^n$. See \cref{lem:tightness of 1.5 johnson} for details.

Our next set of results, which we discuss below, show how to go beyond the above limitation through a new interactive proximity proving technique.

\subsection{Domain Extension for Eliminating Pretenders (DEEP)}\label{sec:deep intro}

The case that interests us most is when $V$ is an RS code (although we will return to the discussion of general linear codes later).
Henceforth, the RS code of rate $\rho$ evaluated over $D$ is  $$\RS[\F,D,\rho]:=\condset{f:D\to \F}{\deg(f)<\rho|D|}.$$
RS codes are maximum distance separable (MDS), meaning that $\delta_V = 1-\rho$ and so \cref{eq:one and a half johnson} simplifies to
\begin{equation}
\label{eq:double johnson RS}
\E[\delta_x] \geq \min(\delta_{\max}, 1-\sqrt[3]{\rho}) - o(1).
\end{equation}
This improved bound
can be translated, using some extra work, to \FRI soundness analysis with similar guarantees. Specifically,  \cref{eq:double johnson RS} implies
that for $f:D\to\F$ that is $\delta$-far from $\RS[\F,D,\rho]$,
the soundness error of a single invocation of the \FRI QUERY test
(which requires $\log |D|$ queries) is at most $\max\{1-\delta,\sqrt[3]{\rho}\}$, and this can be plugged
into ZK-STARKs like~\cite{stark} and ZK-SNARGs like Aurora~\cite{Aurora}. Roughly speaking, if the
rejection probability is of $\delta$-far words is $\max(\delta,\delta_0)$ then
to reach soundness error less than $2^{-\lambda}$ for codes of blocklength $n$,  communication complexity (and verifier complexity) scale roughly like $\frac{\lambda}{\log \delta_0}\cdot c\cdot \log n$ for some constant $c$.
Thus, the improvement from \cref{eq:double johnson} to \cref{eq:one and a half johnson} translates to a $25\%$ reduction in verifier complexity (from $\frac{4\lambda}{\log \rho}\cdot c\cdot \log n$ to $\frac{3\lambda}{\log \rho}\cdot c\cdot \log n$).

To break the soundness bound of \cref{eq:one and a half johnson} and thereby further reduce verifier complexity in the afore-mentioned systems, we suggest a new method. We discuss it first for RS codes, then generalize to arbitrary linear codes.
If $u^*,u:D\to\F$ are indeed the evaluation of two
degree $d$ polynomials, say, $P^*$ and $P$, our verifier will artificially \emph{extend the domain} $D$ to a larger one $\bar{D}$,
sample uniformly $z\in \bar{D}$ and ask for the evaluation of
$P^*(z)$ and $P(z)$. The answers provided by
the prover can now be applied to modify each of $u^*$ and $u$ in a \emph{local} manner to reflect the new knowledge, and along the way also
prune down the large list of polynomials which $u^*$ and $u$ might pretend to be. If $\alpha^*_z=P^*(z),\alpha_z=P(z)$ are the honest prover's answers to the query $z$, then $(X-z)$ divides $P^*(X)-\alpha^*_z$ and likewise $(X-z)|P(X)-\alpha_z$. Letting $\alpha_x=\alpha^*+x\alpha$ and $P_x(X)=P^*(X)+xP(X)$ it follows that $(X-z)|P_x(X)-\alpha_x$.
Consider now the soundness of this procedure.
In the extreme case that $u^*$ has a small list of polynomials that, each, somewhat agree with it, then with high probability over $z$, any answer provided by the prover will agree with at most one of the polynomials in this list. The proof of our main technical result, \cref{lem:main}, formalizes this intuition.
For radius $\delta$, let $L^*_\delta$ be the maximal list size,
$$L^*_\delta=\max_{u^*\in \F^D}|\condset{v\in V}{\Delta(u^*,v)<\delta}|$$
where $\Delta$ denotes relative Hamming distance.
Let $V|_{u_x(z)=\alpha_x}$ be the restriction of $V$ to
codewords that are evaluations of polynomials of degree at most $d$ that, additionally, evaluate to $\alpha_x$ on $z$. Our main
 \cref{lem:main} shows that if $\Delta(u^*,V)=\delta_{\max}$ then for any pair of answers $\alpha^*_z,\alpha_z$ given in response to query $z$,
\begin{equation}
\label{eq:new soundness bound}
\E_{z,x}\left[\Delta(u_x,V|_{u_x(z)=\alpha_x})\right]\geq \delta_{\max}-L^*_\delta\cdot\left(\frac{\rho|D|}{|\bar{D}|}\right)^{1/3}-o(1).
\end{equation}

The Johnson bound (\cref{thm:johnson}) says that when $\delta<J(1-\rho)=1-\sqrt{\rho}$ we
have $L^*_\delta=O(1)$  and this improves
the worst-case-to-average-case result from that of \cref{eq:one and a half johnson} to a bound that matches the Johnson bound:
\begin{equation}
\label{eq:johnson bound deep fri}
\E_{z,x}[\Delta(u_x,V|_{u_x(z)=\alpha_x})|] \geq \min\left(\delta_{\max}, J(\delta_V)\right)-o(1) = \min\left(\delta_{\max}, \sqrt{\rho}\right)-o(1).
\end{equation}
The exact behavior of the list size of Reed-Solomon codes beyond the Johnson bound is a
famous open problem. It may be the case that the list size is small for radii far greater than the Johnson bound; in fact, for most domains $D$ this is roughly known to hold~\cite{RudraW14}. If it holds that
that list sizes are small all the way up to radius equal to the distance $\delta_V = 1-\rho$ (i.e., if Reed-Solomon codes meet list-decoding capacity), then \cref{eq:johnson bound deep fri} implies that the technique
suggested here has optimal soundness for (nearly) all distance
parameters.

\paragraph{Generalization to arbitrary linear codes}
The DEEP method can be used to improve worst-to-average-case reductions for general linear codes. Viewing codewords in $V$ as evaluations of linear forms of a domain $D$, we ask for the evaluation of the linear forms that supposedly correspond to $u^*$ and $u$ on a random location $z\in \bar{D}$ where $|\bar{D}|\gg |D|$. \cref{lem:main general codes} generalizes \cref{lem:main} and says that if $V$ has near-capacity  list-decoding radius (with small list size) and $\bar{D}$ corresponds to (columns of a generating matrix of) a good error correcting code, then we have
$\E_{x}[\delta_x]\approx \delta_{\max}$.
The main difference between the RS case and that of general linear codes is that in the former, the prover-answers $\alpha^*(x), \alpha(x)$ can be processed to modify locally the entries of $u_x$ to reduce the degree of the resulting
function; this is something we cannot carry out (to best of our understanding) for all linear codes.

\subsection{\DFRI}\label{sec:deep fri intro}

Applying the technique of domain extension for eliminating pretenders to
the \FRI protocol requires a modification that we discuss next. The \FRI protocol can be described as a process of ``randomly folding'' an (inverse) Fast Fourier Transform (iFFT) computation. In the  ``classical'' iFFT, one starts with a function $f\zr:\langle \omega \rangle\to \F$ where $\omega$ generates
a multiplicative group of order $2^k$ for integer $k$.
The iFFT computes (in arithmetic complexity $O(k 2^k)$)
the interpolating polynomial $\tilde{f}(X)$ of the function $f$.
This computation follows by computing (in linear time)
a pair of functions $f_0,f_1:\langle{\omega^2}\rangle\to\F$, recalling $|\langle\omega^2\rangle|=\frac{1}{2}|\langle\omega\rangle|$. Their interpolants $\tilde{f}_0,\tilde{f}_1$ are then used to compute in linear time the  original interpolant $\tilde{f}$ of $f$.

As explained in~\cite{FRI}, in the \FRI protocol the prover first commits
to $f$ as above. Then the verifier samples a random $x\zr\in\F$ and
the protocol continues with the single function $f\one:\langle\omega^2\rangle\to\F$ which is supposedly $f\one:=f_0+x\zr f_1$.
It turns out that if $f$ is indeed of degree less than $\rho|\langle \omega \rangle|$ then for all $x$ we have that $f\one$ is of degree
less than $\rho|\langle \omega^2 \rangle|$ as well. The tricky part is showing
that when $f$ is $\delta$-far from $\RS[\F,\langle \omega \rangle,\rho]$ this also holds with high probability (over $x$) for $f\one$ and
some $\delta'$ that is as close as possible to $\delta$. (One can show that invariably we have $\delta'\leq \delta$, i.e., the green line of \cref{fig:FRI} is an upper bound on soundness of both \FRI and the new \DFRI protocol described below.)

The worst-case-to-average-case results of \cref{eq:one and a half johnson} and \cref{lem:one and half Johnson} can be converted to similar improvements for \FRI, showing that
for $\delta<1-\sqrt[3]{\rho}$ we have $\delta'\approx \delta$. This
follows directly from the techniques of \cite[Section~7]{Ben-SassonKS18} (see the red line in \cref{fig:FRI}).
But to use the new DEEP technique of \cref{eq:new soundness bound} and \cref{lem:main} in order to improve soundness of an RS-IOPP, we need to modify the \FRI protocol, leading to a new protocol that is aptly called \DFRI. Instead of constructing $f\one$ directly, our verifier first samples $z\zr\in\F$ and queries the prover
for the evaluation of the interpolant of $f\zr$ on $z\zr$ and $-z\zr$.
After the answers $\alpha_{z\ii},\alpha_{-z\ii}$ have been recorded,
the verifier proceeds by sampling $x\zr$ and expects the prover
to submit $f\one$ which is the linear combination of $f'_0,f'_1$ derived from the modification $f'$ of $f$ that takes into account the answers $\alpha_{z\ii},\alpha_{-z\ii}$. Assuming $\tilde{f}$ is the interpolant
of $f$, an honest prover would set $f'(X):=(\tilde{f}(X)-U(X))/(Z(X))$
where $U(X)$ is the degree $\leq 1$ polynomial that evaluates
to $\alpha_{z\zr}$ on $z\zr$ and to $\alpha_{-z\zr}$ on $-z\zr$ and
$Z(X)$ is the monic degree $2$ polynomial whose roots are $z\zr$ and $-z\zr$.
As shown in \cref{sec:deep fri}, the soundness bounds of \cref{eq:new soundness bound} and \cref{lem:main} now apply to \DFRI. This shows that
the soundness of \DFRI, i.e., the rejection probability of words that are $\delta$-far from $\RS[\F,D,\rho])$ is roughly $\delta$ for any $\delta$ that is smaller than the maximal radius for which list-sizes are ``small''. \cref{fig:FRI} summarizes the results described here.

\definecolor{spgreen}{RGB}{46,139,87}
\begin{figure}[ht]
	\begin{tikzpicture}[scale=0.8]
	\begin{axis}[domain=0.000001:1, samples=100,grid=major,
	restrict y to domain=0:1,xlabel={$\Rate$},ylabel={$\delta_0$}, legend pos= outer north east, legend cell align={left}]
	\addplot [line width=0.25mm,color=spgreen]  {1-x};
	\addplot [line width=0.25mm,color=blue] {1-sqrt(x)};
	\addplot [line width=0.25mm,color=red]    {1-x^(1/3)};
	\addplot [line width=0.25mm,color=orange, dashed]    {1-sqrt(sqrt(x))};
	\addplot [line width=0.25mm,color=red,dashed] {(1-3*x)/4};

	\legend{upper bound + \DFRI conjectured lower bound (\cref{conj:L}), \DFRI lower bound (\cref{thm:DFRI soundness}), \FRI lower bound based on \cref{lem:one-and-a-half-positive}, \FRI previous lower bound \cite{Ben-SassonKS18}, \FRI initial lower bound \cite{FRI}}
	\end{axis}
	\end{tikzpicture}
	\caption{\FRI and \DFRI soundness threshold $\delta_0$ as a function of RS code rate $\Rate$, for a single invocation of the \query phase, as field size $q\to\infty$.
	$\delta_0(\rho)$ is defined to be the largest distance parameter $\delta$ for which soundness (rejection probaiblity) of a single invocation
	of the \FRI/\DFRI QUERY is $\delta-o(1)$.
	Higher lines are better. The top line is the trivial upper bound on soundness which applies to both \FRI and \DFRI; the bottom line is the soundness of the original analysis of \cite{FRI}. Dashed lines represent prior results. The red line is the (tight) soundness lower bound for \FRI and the blue line is a lower bound on \DFRI soundness. Under a plausible conjecture for Reed-Solomon list-decodability (\cref{conj:L}), the actual soundness is as high as the green line. }
	\label{fig:FRI}
\end{figure}

\subsection{DEEP Algebraic Linking IOP (\DEEPSTIK)}\label{sec:deepstik intro}

In \cref{sec:deep fri intro} we only discussed results improving the soundness of Reed-Solomon Proximity Testing (RPT). We now discuss how to improve the soundness of IOP-based argument systems (such as~\cite{stark,Aurora}) that use RPT solutions.
In order to reap the benefits of the improved soundness of RPT,
we need reductions that produce instances of the RPT problem
{\em that are very far from the relevant RS code} when the input instance is
unsatisfiable. One such protocol is the Algebraic Linking IOP (\ALI) of \cite{stark}. 
The instances of the RPT problem derived from an unsatisfiable instance in \ALI are proven to be somewhat far from low-degree but the distance bound proved in that paper is less than $1/8$, even when used with RS codes of negligible rate $\rho$ (nevertheless it is conjectured and assumed in both ZK-STARK  \cite{stark} and Aurora \cite{Aurora} that the distance is significantly greater). In \cref{sec:Deep stik} we use the DEEP technique to modify the \ALI protocol in a manner similar to the \DFRI modification.
The result of this modification allows us to apply the soundness results of
\cref{eq:new soundness bound}
to the \DEEPSTIK protocol and show that, when provided with unsatisfiable instances, the
distance of the received words that result from that protocol is provably at least $1-\sqrt{\rho}-o(1)$ (and may be greater, assuming more favourable bounds on the list decoding radius for RS codes, as in  \cref{conj:L}).

\paragraph{Organization of the rest of the paper}
\cref{sec:preliminaries} presents general notation. \cref{sec:general codes} gives an improved worst-to-average case reductions for general spaces and shows that the bound in the reduction is tight (\cref{lem:tightness of 1.5 johnson}).
\cref{sec:main lemma} presents our main technical result, showing that the DEEP method improves worst-to-average case reductions for RS codes up to the Johnson bound (provably) and perhaps even beyond. \cref{sec:deep fri} presents the \DFRI protocol that obtains better soundness than the state of the art \FRI protocol, and \cref{sec:Deep stik} discusses the \DEEPSTIK protocol.

\section{Preliminaries}\label{sec:preliminaries}

\paragraph{Functions} For a set $D$, we will be working with the space of functions $u: D \to \F$, denoted $\F^D$.
For $u\in \F^D$ we use $u(z)$ to denote the $z$th entry of $u$, for $z\in D$. For $C\subset D$ we use $f|_C$ to denote the restriction of $f$ to $C$. For two functions $f,g:D\to\F$ we write $f=g$ when the two functions are equal as elements in $\F^D$ and similarly say $f|_C=g|_C$ when their restrictions are equal as elements in $\F^C$.

\paragraph{Distance} We use $\Delta_D(u,v)=\Pr_{z\in D}\left[u(z)\neq v(z)\right]$ for relative Hamming distance, and omit $D$ when it is clear from context.
For a set $S\subset \F^D$ we use $\Delta_D(v,S)=\min_{s\in S}\Delta_D(v,s)$ and $\Delta_D(S)=\min_{s\neq s'\in S}\Delta_D(s,s')$ denotes the minimal relative distance of $S$.
For $u\in\F^D$ let $\ball(u,\delta)$ denote the Hamming
ball in $\F^D$ of normalized radius $\delta$ centered at $u$,
\[\ball(u,\delta)=\condset{u'\in\F^D}{\Delta_D(u,u')<\delta}.\]

\paragraph{Linear codes} An $[n,k,d]_q$-linear error correcting code is a linear space $V\subset\F_q^{n}$ of dimension $k$ over $\F_q$ with minimal Hamming distance $d$. A generating matrix for $V$ is a matrix
$G\in \F_q^{n\times k}$ of rank $k$ such that $V=\condset{Gx}{x\in \F_q^k}$.

\paragraph{Polynomials and RS codes} The \emph{interpolant} of $f:D\to\F_q$ is the unique polynomial of degree $<|D|$ whose evaluation on $D$ is $f$. The degree of $f$, denoted $\deg(f)$, is the degree of its interpolant.
The RS code evaluated over domain $D\subset \F$ and rate $\rho$
is denoted $\RS[\F,D,\rho]=\condset{f:D\to\F}{\deg(f)<\rho|D|}$. Sometimes it will be more convenient to work with degree rather than rate, in which case we abuse notation and define
$\RS[\F,D,d]=\condset{f:D\to\F}{\deg(f)< d}$.
We use capital letters like $P,Q$ to denote polynomials and
when we say $P\in \RS[\F,D,\rho]$ we mean that $\deg(P)<\rho|D|$ and
associate $P$ with the RS codeword that is its evaluation on $D$. We also use $\tilde{f}$ to denote the interpolant
of a function $f$.

\subsection{List Decoding}\label{sec:list decoding}

\begin{definition}[List size for Reed-Solomon Codes]\label{def:list size}
	For $u\in \F^D$, a set $V\subset \F^D$, and distance parameter $\delta\in [0,1]$,
	let $\Ls(u,V,\delta)$ be the set of elements in $V$ that are at most $\delta$-far
	from $u$ in relative Hamming distance. Formally, using $B(u,\delta)$ to denote the
	Hamming ball of relative radius $\delta$ centered around $u$, we have
	$\Ls(u,V,\delta)=B(u,\delta)\cap V$.

	The code $V$ is said to be {\em $(\delta, L)$-list-decodable} if
	$|\Ls(u, V, \delta)| \leq L$ for all $u\in \F_q^D$.

	For $D \subseteq \F_q$, let $\calL(\F_q, D, d, \delta)$ be the maximum size of
	$\Ls(u,V,\delta)$ taken over all $u\in \F_q^D$ for $V=\RS[\F_q,D,\rho=d/|D|]$.
\end{definition}



We recall the fundamental Johnson bound, which says that sets with large minimum distance
have nontrivial list-decodability. The particular
version below follows, e.g., from \cite[Theorem 3.3]{guruswami2007algorithmic} by
setting $d = (1-\rho)|D|$ and $e = (1-\sqrt\rho-\johnsoneps)|D|$ there.


\begin{theorem}[Johnson bound]
\label{thm:johnson}

Let $V\subset \F^D$ be a code with minimum relative distance
$1-\rho$, for $\rho \in (0, 1)$.
Then $V$ is
$(1 - \sqrt{\rho} - \johnsoneps, 1/(2\johnsoneps\sqrt{\rho}))$-list-decodable
for every $\johnsoneps \in (0, 1 - \sqrt{\rho})$.
\end{theorem}

In particular, for Reed-Solomon codes this implies the following list-decodability bound:
$$ \calL( \F_q, D, d  = \rho |D|, 1 - \sqrt{\rho} - \johnsoneps) \leq O\left(\frac{1}{\johnsoneps \sqrt{\rho}}\right).$$

Extremely optimistically, we could hope that Reed-Solomon codes
are list-decodable all the way up to their distance with moderate list sizes.
Staying consistent with the known limitations~\cite{BKR-RSLD}, we have the
following brave conjecture.
\begin{conjecture}[List decodability of Reed-Solomon Codes up to Capacity]
	\label{conj:L}
	For every $\rho > 0$, there is a constant $C_\rho$ such that
 every Reed-Solomon code of length $n$ and rate $\rho$ is list-decodable from $1-\rho-\varepsilon$ fraction errors with list size $\left( \frac{n}{\varepsilon} \right)^{C_\rho}$. That is:
$$ \calL( \F_q, D, d = \rho |D|, 1 - \rho - \varepsilon) \leq \left(\frac{|D|}{\varepsilon} \right)^{C_\rho}.$$
\end{conjecture}

\section{Improved High-error Distance Preservation}
\label{sec:general codes}

Our first result gives better distance preservation results for linear codes $V$ of relative distance $\lambda$.
The previous state-of-the-art \cite{Ben-SassonKS18}
said that when a $1$-dimensional affine space $U$ contains some element $u^*$ that is $\delta_{\max}=\Delta(u^*,V)$ far from $V$,
then $$\E_{u\in U}[\Delta(u,V)]\geq \min(\delta_{\max},1-J\two(\lambda))-o(1).$$
The following lemma improves the average-case distance to
$$\E_{u\in U}[\Delta(u,V)]\geq \min(\delta_{\max},1-J\ohalf(\lambda))-o(1).$$
Later on, in \cref{sec:tightness of one and half johnson}, we will show that this result is tight (for a sub-family of RS codes).

\begin{lemma}[One-and-half Johnson distance preservation]\label{lem:one and half Johnson}
	Let $V \subseteq \F_q^n$ be a linear code of distance $\lambda=\Delta(V)$.
	Let $\epsilon, \delta > 0$ with $\epsilon < 1/3$ and  $\delta < 1 - (1-\lambda + \epsilon)^{1/3}$.

	Suppose $u^* \in \F_q^n$ is such that $\Delta(u^*, V) > \delta + \epsilon$.
	Then for all $u \in \F_q^n$, there are at most $2/\epsilon^2$
	values of $x \in \F_q$ such that $\Delta(u^* + x u, V) < \delta$.
\end{lemma}

This result is the contra-positive statement of the following, more informative, version of it, that we prove below.

\begin{lemma}[One-and-half Johnson distance preservation --- positive form]\label{lem:one-and-a-half-positive}
Let $V \subseteq \F_q^D$ be a linear code of distance $\lambda=\Delta(V)$.
Let $\epsilon, \delta > 0$ with $\epsilon < 1/3$ and $\delta < 1 - (1-\lambda + \epsilon)^{1/3}$.
	Let $u,u^* \in \F_q^D$ satisfy
	\begin{equation}\label{eq1.5}
	\Pr_{x \in \F_q} [ \Delta( u^* + x u, V) < \delta ] \geq  \frac{2}{\epsilon^2 q}.
	\end{equation}
	Then there exist $v, v^*\in V$ and $C \subseteq D$ such that the following three statements hold simultaneously:
\begin{itemize}
\item $|C| \geq (1-\delta-\epsilon) |D|$,
\item $u|_{C}=v|_{C}$, and
\item $u^*|_{C}=v^*|_{C}$.
\end{itemize}
\end{lemma}
Observe that if $u, u^*$ satisfy \cref{eq1.5} then the $v, v^*,C$ deduced by \cref{lem:one-and-a-half-positive} have the property that for {\em all} $x \in \F_q$, we have $\Delta(u^* + xu, V) \leq \delta + \epsilon$. In other words, the existence
of $v, v^*$ and $C$ almost completely explains Equation~\eqref{eq1.5}.

Quantitatively weaker statements in this vein were proved by~\cite{PS94, FRI} in the low-error case, 
and~\cite{ChiesaMS17,Ben-SassonKS18} in the high-error case.  The proofs of the latter two results
used combinatorial tools (the Kovary-Sos-Turan bound and the Johnson bound respectively) that
are closely related to one another. Our improved  proof below is direct, and is based on  the same
convexity principle that underlies both the Kovary-Sos-Turan and Johnson bounds.
\begin{proof}
Let $u_x = u^* + x u$.
Let $$A = \{x \mid \Delta(u^* + x u, V) < \delta \}.$$
For each $x \in A$, let $v_x \in V$ be an element of $V$ that is closest
to $u_x$, and let $S_x \subseteq D$ be the agreement set of $u_x$ and $v_x$, defined as $S_x=\condset{y\in D}{u_x(y)=v_x(y)}$.

For $x, \beta, \gamma$ picked uniformly from $A$ and $y$ picked uniformly from $D$, we have:
\begin{align*}
\E_{x, \beta,\gamma} [ |S_{x} \cap S_\beta \cap S_\gamma|/n ] &= \E_{y, x, \beta, \gamma} [ 1_{y \in S_x \cap S_\beta \cap S_{\gamma}} ]\\
&= \E_{y} [ \E_{x} [ 1_{y \in S_{x}}]^3 ]\\
&\geq \E_{y, x} [ 1_{y \in S_{x}} ]^3 \\
&\geq (1-\delta)^3\\
&> 1-\lambda + \epsilon.
\end{align*}

The second equality above follows from the independence of the events
$y\in S_x, y\in S_\beta, y\in S_\gamma$ given $y\in D$. The first inequality is Jensen's and the last inequality is by assumption on $\delta, \gamma,\epsilon$.

Thus $$\Pr_{x, \beta, \gamma} [ |S_{x} \cap S_\beta \cap S_\gamma| \geq (1-\lambda)|D| ] \geq \epsilon .$$

Note that $\Pr_{x,\beta,\gamma}[x,\beta,\gamma \text{ are not all distinct}] < 3/|A|$.
Since $|A| \geq 2/\epsilon^2 >  \frac{6}{\epsilon}$, we have that $3/|A| \leq \epsilon/2$ and hence
$x, \beta, \gamma$ are all distinct with probability at least $1 -\epsilon/2$.
Thus with probability at least $\epsilon/2$ over the choice of $x, \beta, \gamma$, we have that $x, \beta, \gamma$ are all distinct and
$|S_x \cap S_\beta \cap S_\gamma| > (1-\lambda)|D|$.

This means that there are distinct $x_0, \beta_0$
such that
$$ \Pr_{\gamma} [  |S_{x_0} \cap S_{\beta_0} \cap S_\gamma| > (1-\lambda)|D|] \geq \epsilon/2.$$

Fix a $\gamma$ where this happens. Let $S = S_{x_0} \cap S_{\beta_0} \cap S_{\gamma}$.
We have that $$(x_0, u_{x_0}), (\beta_0, u_{\beta_0}), (\gamma, u_\gamma)$$
are collinear.
Thus $$(x_0, u_{x_0}|_S), (\beta_0, u_{\beta_0}|_S), (\gamma, u_{\gamma}|_S)$$
are all collinear.
By definition of $S$, we get that:
$$(x_0, v_{x_0}|_S), (\beta_0, v_{\beta_0}|_S), (\gamma, v_{\gamma}|_S)$$
are all collinear.
Since $|S| > (1-\lambda)|D|$ (and recalling that $\lambda$ is the distance of $V$), we get that $v_\gamma$ is determined by $v_\gamma|_S$ via a linear map. This means that
$$(x_0, v_{x_0}), (\beta_0, v_{\beta_0}), (\gamma, v_{\gamma})$$
are all collinear.

Thus $\epsilon/2$-fraction of the $\gamma \in A$ have the ``good" property that
$(\gamma, v_{\gamma})$ is on the line passing through
$(x_0, v_{x_0})$ and $(\beta_0, v_{\beta_0})$.
Write this line as $v^* + x v$ and notice that for all ``good'' $\gamma$ we have $v_{\gamma} = v^* + \gamma v$.
Let $A'\subseteq A$ denote the set of good elements for this line, recording that $|A'|\geq |A|\cdot \epsilon/2\geq 1/\epsilon$.

Thus for $x \in A'$, $\Delta(u^* + xu, v^* + xv) < \delta $.

Consider the set $C\subset D$ defined by
	$$C = \condset{y\in D}{u^*(y)=v^*(y) \mbox{ AND } u(y)=v(y)}.$$

For each $y \in D\setminus C$ there exists at most a single value of $x\in \F_q$ satisfying $u^*(y) + x\cdot u(y)=v^*(y) + x\cdot v(y)$ because
	$$(u^*(y)-v^*(y)) + x\cdot(u(y)-v(y))$$
	has at most one value $x$ on which it vanishes.

This implies
	$$\delta \geq \E_{x\in A'}[\Delta_D(u_x,v_x)]\geq
	\frac{|D\setminus C|}{|D|}\cdot \left(1-\frac{1}{|A'|}\right)\geq \left(1-\frac{|C|}{|D|}\right)\cdot(1-\epsilon)\geq 1-\frac{|C|}{|D|}-\epsilon.$$
	Rearranging, we get $\frac{|C|}{|D|}\geq 1-(\delta+\epsilon)$ and this completes the proof.

\end{proof}

\subsection{Tightness of the one-and-a-half Johnson bound}\label{sec:tightness of one and half johnson}

\cref{lem:one and half Johnson} says that when $V$ is a linear code
with minimum distance $\lambda$, and $u^*$ is some element that is $\delta$-far from $V$, then for any $u$ we have with high probability
$$\Delta(u^* + x u, V) \geq \min(\delta, J\ohalf(\lambda)= 1-(1-\lambda)^{1/3}).$$

The rightmost term seems quite strange, as the $J\ohalf(\cdot)$ function is
unfamiliar in other settings of coding theory. However, as we show next,
in certain  settings this function gives the correct bound!

\begin{lemma}[Tightness of one-and-a-half Johnson bound]
	\label{lem:tightness of 1.5 johnson}
	For every member $V_n$ of following family of RS codes
	$\condset{V_n=\RS[\F_{2^n},\F_{2^n},\rho=2^{-3}]}{n\in \N}$ there exist $u^*_n,u_n\in\F_{2^n}^{\F_{2^n}}$ satisfying the following:
		\begin{itemize}
			\item $\delta_{\max}\eqndef\Delta(u^*_n,V_n)=\Delta(u_n,V_n)=\frac{3}{4}=1-\rho^{2/3}$
			\item $\forall x\neq 0, \Delta(u^*_n+xu_n,V_n)\leq \frac{1}{2}=1-\rho^{1/3}=J\ohalf(\Delta(V_n))$
		\end{itemize}
	Consequently, $\E[\delta_x]\leq J\ohalf(V_n)+o(1)\leq \delta_{\max}-\frac{1}{4}+o(1)$.
\end{lemma}

We shall need to following claim in our proof of the lemma.

\begin{claim}\label{clm:tightness of one and half johnson}
	For every $x \in \F_{2^n} \setminus\{0\}$ there exists
	a polynomial $P_{x}(Y) \in \F_{2^n}[Y]$ of the form
	$$ P_x(Y) = Y^{2^{n-1}} + x Y^{2^{n-2}} + \tilde{P}_x, \quad \deg(\tilde{P}_x) < 2^{n-3}.$$
	that has $2^{n-1}$ distinct roots in $\F_{2^n}$.
\end{claim}
\begin{proof}
	For $x\neq 0$ let $\beta_x = 1/x^2$, noticing $\beta_x$ is unique because the map $\beta\mapsto \beta^2$ is bijective on $\F_{2^n}$.
	Let $\Tr(Z)\eqndef \sum_{i=0}^{n-1}  Z^{2^i}$ be the trace function
	from $\F_{2^n}$ to $\F_2$. Define
	$$S_x = \{ y \in \F_{2^n} \mid \Tr(\beta_x y) = 0\}.$$
	It is well known that $|S_x| = 2^{n-1}$ because the trace function has $2^{n-1}$ roots in $\F_{2^n}$. So we define
	$$P_x(Y) = \frac{1}{\beta_x^{2^{n-1}}} \cdot \Tr(\beta_x Y) = Y^{2^{n-1}} + \frac{1}{\beta_x^{2^{n-2}}} Y^{2^{n-2}} +\tilde{P}_x = Y^{2^{n-1}} + x Y^{2^{n-2}} + \tilde{P}_x, \quad \deg(\tilde{P}_x(Y)) < 2^{n-3}.$$
	The last equality follows because $\beta_x^{2^{n-2}}=x$.
\end{proof}

\begin{proof}[Proof of \cref{lem:tightness of 1.5 johnson}]
	Consider $V_n$ in this family and let $\F=\F_{2^n}$. Define
	$u^*: \F \to \F$ to be the function $u^*(y) = y^{2^{n-1}}$ and
	let $u:\F \to \F$ be the function $u(y) = y^{2^{n-2}}$.

By \cref{clm:tightness of one and half johnson}, for every $x \in \F \setminus \{0 \}$ there is some $v_x \in V_n$ and $P_x$ with $2^{n-1}$ roots in $\F$ such that
$$ P_{x} - (u^* + x u) + v_x = 0.$$
Then
$$\Delta(u^* + x u, v_x) = \Pr_{y \in \F} [u^*(y) + x u(y) \neq v_x(y)] = \Pr_{y}[ P_x(y) \neq 0 ] = 1/2.$$
Thus we get that for all $x \in \F \setminus \{0\}$
$$ \Delta( u^* + x u, V ) \leq 1/2.$$

On the other hand,
$$\Delta(u, V) \geq 3/4,$$
because for all $v \in V_n$, $u-v$ is a polynomial of degree at most $2^{n-2} = |\F| / 4$. This completes the proof.
\end{proof}

\begin{remark}
Since this example is based on Reed-Solomon codes, it also easily translates into a limitation on the soundness of \FRI. In particular, it means that the improvment to the soundness of \FRI given in~\cref{rem:improved fri analysis} is optimal.
\end{remark}


\eli{Question: Do we want to keep the discussion?}
\paragraph{Discussion}
\label{sec:tightness discussion}
\cref{lem:tightness of 1.5 johnson} raises the question of whether the one-and-a-half Johnson bound of \cref{lem:one and half Johnson} is tight for all RS codes, including non-binary fields and evaluation domains that are strict subsets of the ambient field. We point out that the technique used to prove \cref{lem:tightness of 1.5 johnson} deteriorates rapidly even for binary fields, and even when the evaluation domain is an $\F_2$-linear space which resembles the case above.

Indeed, consider an evaluation domain $D\subset \F_{2^n}$ that is a $d+1$-dimensional linear space over $\F_2$, where $n>d+1$. There are $2^{d+1}$ subspaces of dimension $d$ in $V$. For such $U\subset V, \dim(U)=d$, the polynomial $P_U(X)=\prod_{\alpha \in U}(X-\alpha)$ is of the form
$$P_U(Y)=Y^{2^d}+x_U Y^{2^d-1}+\hat{P}_U(Y)$$
which resembles the structure of \cref{clm:tightness of one and half johnson}. Moreover, as was the case there, for $U'\neq U, U'\subset V, \dim(U')=d$ we have $x_U\neq x_{U'}$. This is because $P_U-P_{U'}$ is a non-zero polynomial with $2^{d-1}$ roots, because  $\dim(U\cap U')=d-1$. Thus, we cannot have $x_U=x_{U'}$ as this would imply $\deg(P_U-P_{U'})\leq 2^{d-2}<2^{d-1}$, contradiction.

As in \cref{lem:tightness of 1.5 johnson}, taking $u^*$ to be the evaluation of $Y^{2^d}$ on $D$ and $u$ be the evaluation of $Y^{2^{d-1}}$ on the same space, we conclude there exists a set $A\subset \F, |A|=2^d$, such that for $x\in A$ we have that $u^*+x u$ agrees with some RS codeword of rate $2^{-3}$ on half of the evaluation domain.

However, notice that $|A|/2^n=2^{-(n-d)}$, meaning that the probability of sampling $x\in A$ deteriorates exponentially with the difference $n-d$. Thus the above counterexample fails to rule out an improvement to~\cref{lem:one and half Johnson} when the length of the code $n$ is much smaller than the size of the field $q$.

Conceivably, both \cref{lem:one and half Johnson} and the analysis of \FRI can be improved significantly under the assumption $n \ll q$. This is the case of most importance to practical implementations of STARKs.

\section{The DEEP Theorem --- Using Domain Extension for Eliminating Pretenders (DEEP) and Improving Soundness}\label{sec:main lemma}


We now come to the statement of our improved-soundness distance preservation result. We describe it first for the special case of RS codes. A weighted variant of the theorem is shown in \cref{sec:deep weighted version} because it is used later in the \DFRI protocol (\cref{sec:deep fri}). We end with \cref{sec:deep lemma for general codes} in which we present a general version of the folowing result, that applies to all linear codes.

\subsection{DEEP Theorem for RS codes}
\label{sec: deep lemma for rs codes}

The vectors $u^*, u$ discussed in the previous section are now viewed as functions $u^*, u : D \to \F_q$ and we are interested in the distance of a random linear combination $u_x=u^*+x\cdot u$ from the code $V=\RS[\F_q, D, \rho]$, where $x\in \F_q$ is sampled uniformly. \cref{lem:one and half Johnson} established that if $\max(\Delta(u^*, V), \Delta(u,V)) = \delta_{\max}$, then with high probability (over $x$), the function $u_x$ will have distance at least $\approx \min(\delta_{\max}, 1-\rho^{1/3})$ from $V$.

\cref{lem:one-and-a-half-positive} roughly gets used in the following way in the FRI protocol. There are two functions $u^*, u : D \to \F_q$ and there is a prover who claims that both are evaluations of low degree polynomials. In order to verify this, the verifier uniformly samples $x\in \F_q$ and considers the function $u_x=u^*+x\cdot u$. \cref{lem:one-and-a-half-positive} shows that if any of $u^*, u$ is far from being evaluations of a low degree polynomial, then so is $u^*+x\cdot u$. This then gets exploited in the FRI protocol using FFT type ideas.

We now precede the random process of sampling $x\in \F_q$ with a step of \emph{domain extension}, explained next. Assume a prover claims that both $u$ and $u^*$ are evaluations of low degree polynomials (say $P(Y)$ and $P^*(Y)$). So these polynomials can be evaluated also outside of $D$. Based on this, a verifier first samples $z \in \F_q$ uniformly and asks the prover to reply with two field elements $a^*, a \in\F_q$ which are supposedly equal to $P^*(z), P(z)$, respectively. After receiving these answers, the verifier proceeds as before, sampling uniformly $x \in \F_q$. Then, setting $b=a^*+x\cdot a$, we examine the distance of $u_x$ from the sub-code $V_{z,b}\subset V$ comprised of all members of $V$ whose interpolating polynomial evaluates to $b$ on input $z$. The code $V_{z,b}$ is the additive coset (shifted by $b$) of a low-degree ideal, the ideal generated by $(X-z)$ (cf. \cref{lem:quotient}).

Using the Johnson Bound (\cref{thm:johnson}) we prove that with high probability $u_x$ is at least $\approx\min(\delta_{\max},1-\rho^{1/2})$ far from $V_{z,b}$. Assuming RS codes have a larger list-decoding radius (\cref{conj:L}), we show that with high probability $u_x$ is $\approx \delta_{\max}$-far from $V_{z,b}$ for nearly all values of $\delta_{\max}$. Later, in \cref{sec:deep fri}, we shall use the improved distance preservation to construct the DEEP-FRI protocol for testing proximity to the RS code with improved soundness.

The statement we give below is given more generally in terms of the list size bound $\calL(\F_q, D, d=\rho|D|, \slack)$;
we instantiate it later with the Johnson bound and with \cref{conj:L}.
It is useful to keep in mind that this will be used in a setting where $q$ is much larger than $|D|$ (and hence $d$), and where $L^*_\slack$ is small.

\begin{theorem}[DEEP method for RS codes]\label{lem:main}
	Let $\rho > 0$ and let $V=\RS[\F_q,D,\rho]$.
	For $z,b \in \F_q$, we let
		$$V_{z,b} = \left\{ Q(Y)|_D \in V \mid Q(z) = b \right\}.$$
	For $\slack>0$ 	let
	$L^*_\slack = \calL(\F_q, D, d=\rho|D|, \slack)$.

	Let $u,u^* \in \F_q^D$.
	For each $z\in \F_q$, let $B_z(X) \in \F_q[X]$ be an arbitrary linear function.
	Suppose that for some $1/3> \epsilon>0$ the following holds,
	\begin{equation}\label{eqDEEP}
	\Pr_{x, z \in \F_q} [ \Delta( u^* + x u, V_{z, B_z(x)}) < \slack ] \geq
	 \max\left(2L^*_\slack \left( \frac{d}{q} + \epsilon\right)^{1/3} , \frac{4}{\epsilon^2 q} \right),
	\end{equation}
	Then there exist $v, v^*\in V$ and $C \subset D$ such that:
\begin{itemize}
\item $|C| \geq (1-\slack-\epsilon) |D|$,
\item $u|_{C}=v|_{C}$,
\item $u^*|_{C}=v^*|_{C}$.
\end{itemize}
Consequently, we have $\Delta(u, V), \Delta(u^*, V) \leq \slack + \epsilon$.
\end{theorem}

\eli{Would be nice to instantiate with a few examples, as we do for the general case}

\begin{proof}
	To simplify notation set $\eta = \max\left( 2L^*_\slack \left( \frac{d}{q} + \epsilon\right)^{1/3}, \frac{4}{\epsilon^2 q} \right)$,
and let $u_x = u^* + x u$.

Let $\event[x,z]$ denote the event ``$\exists P(Y) \in \Ls(u_x,V,\slack), P(z)=B_z(x)$''.

The assumption of Equation \eqref{eqDEEP} now reads as
	$$\Pr_{x,z\in\F_q}[\event[x,z]]\geq \eta.$$

	Thus we get,
	\begin{align}
	\label{zxmarkov}
	\Pr_{x} [ \Pr_{z} [\event[x,z]]\geq \eta/2] \geq \eta/2
	\end{align}

	Let $$A = \condset{x\in \F_q}{\Pr_{z} [\event[x,z]]\geq \eta/2]} $$
	and notice $|A|\geq \eta q/2$.

    For $x \in \F_q$, pick $P_x \in V$ to be a member $P\in \Ls(u_x,V,\slack)$ that maximizes
    $\Pr_{z \in \F_q} [ P(z) = B_z(x)]$.
	Let $S_x = \condset{z \in \F_q}{P_x(z) = B_z(x)}$ and set
	$\mu_x = |S_x|/q$.
	By definition, $|\Ls(u_x,V,\slack)| \leq L^*_\slack$, and so by the pigeonhole principle, for each $x \in A$ we have
	$\mu_x\geq \frac{\eta}{2L^*_\slack}$.

	For $x, \beta, \gamma$ picked uniformly from $A$, and $z$ picked uniformly from $\F_q$, we have:
	\begin{align*}
	\E_{x, \beta,\gamma} [ |S_{x} \cap S_\beta \cap S_\gamma|/q ] &= \E_{z, x, \beta, \gamma} [ 1_{z \in S_x \cap S_\beta \cap S_{\gamma}} ]\\
	&= \E_{z} [ \E_{x} [ 1_{z \in S_{x}}]^3 ]\\
	&\geq \E_{z, x} [ 1_{z \in S_{x}} ]^3 \\
	&\geq \left(\frac{\eta}{2L^*_\slack}\right)^3\\
	&> \frac{d}{q} + \epsilon.
	\end{align*}

	The second equality above follows from the independence of
	$x, \beta, \gamma$. The first inequality is an application of Jensen's inequality and the last inequality is by assumption on $\eta$.

	Thus $$\Pr_{x, \beta, \gamma} [ |S_{x} \cap S_\beta \cap S_\gamma| > d ] \geq \epsilon .$$

	Note that $\Pr_{x,\beta,\gamma}[x,\beta,\gamma \text{ are not all distinct}] < 3/|A|$.
	Since $|A| \geq \eta q/2 \geq 2/\epsilon^2 \geq 6/\epsilon$ we have $3/|A| \leq \epsilon/2$.
	Thus $\Pr_{x,\beta,\gamma}[x,\beta,\gamma \text{ are all distinct and }|S_x \cap S_\beta \cap S_\gamma| > d] \geq \epsilon/2$.

	This means that there are distinct $x_0, \beta_0$
	such that
	$$ \Pr_{\gamma} [  |S_{x_0} \cap S_{\beta_0} \cap S_\gamma| > d] \geq \epsilon/2.$$

	Consider some $\gamma$ where this happens. Let $S = S_{x_0} \cap S_{\beta_0} \cap S_{\gamma}$.
	By construction we know that for all $z \in \F_q$,
	$$(x_0, B_z(x_0)), (\beta_0, B_z(\beta_0)), (\gamma, B_z(\gamma))$$
	are collinear. So, in particular, for $z \in S$ this holds.

	By definition of $S$, we get that for each $z \in S$,
	$$(x_0, P_{x_0}(z)), (\beta_0, P_{\beta_0}(z)), (\gamma, P_{\gamma}(z)) \in \F_q \times \F_q$$
	are collinear.
	Since $|S| > d$, we have that $P_\gamma$ is uniquely determined by $P_\gamma|_S$ by a linear map. This allows us to conclude that
	$$(x_0, P_{x_0}), (\beta_0, P_{\beta_0}), (\gamma, P_{\gamma}) \in \F_q \times \F_q[Y]$$
	are collinear in the $\F_q$-vector space $\F_q \times \F_q[Y]$.

	Thus, an $\epsilon/2$-fraction of the $\gamma \in A$ have the ``good'' property that
	$(\gamma, P_{\gamma})$ is on the line passing through
	$(x_0, P_{x_0})$ and $(\beta_0, P_{\beta_0})$.
    Write this line as $P^* + x P$ and notice that for all ``good'' $\gamma$ we have $P_\gamma=P^*+\gamma P$.
	Let $A'\subseteq A$ denote the set of good elements for this line, recording that $|A'|\geq |A|\cdot \epsilon/2\geq 1/\epsilon$.
	By definition of $\Ls(u_x,V,\slack)$ and
	the assumption $P_x\in \Ls(u_x,V,\slack)$,
	we have that $\Delta(u_x, P_x) < \slack$ for $x\in A'$.

	Consider the set $C\subset D$ defined by
	$$C = \condset{y\in D}{u^*(y)=P^*(y) \mbox{ AND } u(y)=P(y)}.$$

	For each $y \in D\setminus C$ there exists at most a single value of $x\in \F_q$ satisfying $u_x(y)=P_x(y)$ because
	$$u_x(y)-P_x(y)=(u^*(y)-P^*(y)) + x\cdot(u(y)-P(y))$$
	has at most one value $x$ on which it vanishes. This implies
	$$\slack\geq \E_{x\in A'}[\Delta_D(u_x,v_x)]\geq
	\frac{|D\setminus C|}{|D|}\cdot \left(1-\frac{1}{|A'|}\right)\geq \left(1-\frac{|C|}{|D|}\right)\cdot(1-\epsilon)\geq 1-\frac{|C|}{|D|}-\epsilon.$$
	Rearranging, we get $\frac{|C|}{|D|}\geq 1-(\slack+\epsilon)$. Taking $v = P$ and $v^* = P^*$ completes the proof.
\end{proof}

\begin{remark} One could extend the domain even further, and sample $z$ from
an extension field $\F_{q^a}$. This gives even better soundness; the expression
$2 L^*_\slack \cdot \left(\frac{d}{q} + \epsilon \right)^{1/3}$ by
$2 L^*_\slack \cdot \left(\frac{d}{q^a} + \epsilon \right)^{1/3}$.
This can give interesting results even if $L^*_\slack = q^{O(1)}$ by taking $a = O(1)$.
\end{remark}

\subsection{Weighted version}
\label{sec:deep weighted version}

For application to Reed-Solomon Proximity Testing, it is more convenient to have a weighted version of the previous result.
We briefly introduce some notation for dealing with weights, and then state the new version.

Let $u, v \in \F_q^D$.
Let $\eta \in [0,1]^D$ be a vector of weights.
We define the $\eta$-agreement between $u$ and $v$ by:
$$\agree_\eta(u, v) = \frac{1}{|D|} \sum_{i \in D \mid u_i = v_i } \eta(i).$$

For a subspace $V \subseteq \F_q^n$, we define
$$\agree_\eta(u, V) = \max_{v \in V} \agree_\eta(u, v).$$

\begin{theorem}
\label{lem:weighted-OOD-main}
	Let $\rho > 0$ and let $V=\RS[\F_q,D, d = \rho \cdot |D|]$.
	For $z,b \in \F_q$, we let
		$$V_{z,b} = \{ Q(Y) \in V \mid Q(z) = b \}.$$
For $\alpha < 1$, let $L^* = \calL(\F_q, D, d=\rho|D|, 1-\alpha)$  be the list-size for list-decoding $V$
from $(1 - \alpha)$-fraction errors (without weights).

	Let $u,u^* \in \F_q^D$.
	For each $z\in \F_q$, let $B_z(X) \in \F_q[X]$ be an arbitrary linear function.
	Suppose that
	\begin{equation}\label{wtd-eq1}
	\Pr_{x, z \in \F_q} [ \agree_{\eta}( u^* + x u, V_{z, B_z(x)}) > \alpha] \geq
	 \max\left(2L^* \left( \frac{d}{q} + \epsilon\right)^{1/3} , \frac{4}{\epsilon^2 q} \right),
	\end{equation}
	Then there exist $v, v^*\in V$ and $C \subset D$ such that:
\begin{itemize}
\item $\sum_{y \in C} \eta(y) > (\alpha - \epsilon) |D|$,
\item $u|_{C}=v|_{C}$,
\item $u^*|_{C}=v^*|_{C}$.
\end{itemize}
Consequently, we have $\agree_\eta(u, V), \agree_\eta(u^*, V) \geq  \alpha - \epsilon$.
\end{theorem}

The proof is nearly identical to the proof of \cref{lem:main} so we only highlight the changes.
First, we observe that if $\eta_1 : D \to [0,1]$ is the the constant function with value $1$,
then $\agree_\eta(u, v) \leq \agree_{\eta_1}(u,v) = 1 - \Delta(u,v)$.
Thus the set $$\{Q(Y) \in \F_q[Y] \mid \deg(Q) \leq d, \agree_\eta(u^* + x u, Q) > \alpha \}$$
is contained in $$\{ Q(Y) \in \F_q[Y] \mid \deg(Q) \leq d, \Delta(u^* + xu, Q) < 1-\alpha \}.$$
The size of this latter set is bounded by $L^*$, and thus the size of the former set is too.
The proof then proceeds as before, until the very end, where we have a set $A' \subseteq \F_q$,
with $|A'| \geq \frac{2}{\epsilon}$, and polynomials $P, P^* \in V$
such that for each $x \in A'$, $\agree_{\eta}(u^* + xu, P^* + x P) > \alpha$.
Then we take $C = \{ y \in C \mid u^*(y) = P^*(y), u(y) = P(y) \}$, and our goal is to show
that $\sum_{y\in C} \eta(y) > (\alpha - \epsilon) |D|$.
To this end, consider:
\begin{align*}
\alpha &< \frac{1}{|A'|}\sum_{x \in A'} \agree_{\eta}(u^* + xu, P^* + x P)\\
&= \frac{1}{|D| \cdot |A'|}\sum_{x \in A'} \sum_{y\in D} (\eta(y) \cdot 1_{u^*(y) + xu(y) = P^*(y) + xP(y)} )\\
&= \frac{1}{|D|} \sum_{y \in D} \eta(y) \left( \frac{1}{|A'|} \sum_{x \in A'} 1_{u^*(y) + xu(y) = P^*(y) + xP(y)} \right)\\
&\leq \frac{1}{|D|} \sum_{y \in C} \eta(y) + \frac{1}{|D|}\sum_{y \in D \setminus C} \eta(y) \cdot \frac{1}{|A'|}\\
&\leq \frac{1}{|D|}\sum_{y \in C} \eta(y) + \epsilon/2.
\end{align*}
This implies that $\sum_{y \in C} \eta(y) > (\alpha - \epsilon) |D|$, and the rest of the proof is the same as before.

\subsection{DEEP Lemma for general linear codes}
\label{sec:deep lemma for general codes}

\cref{lem:main} can be generalized to apply to arbitrary linear codes, and this is the focus of this section. We explain the basic principles
for an $[n,k,d]_q$-linear code $V$ with generating matrix $G\in \F_q^{ k \times n}$, viewing codewords as evaluations of linear forms on the columns of $G$.

Let $D\subset\F_q^k$ be the set of columns of $G$.
A linear form $\ell \in F_q^k$ can be ``evaluated" at any any element $x$ of $D$.
Similarly, if we fix a set of points $S \subseteq \F_q^k$ (thinking $|S| \gg |D|$),
we may evaluate the linear form $\ell$ at any point of $S$ -- this corresponds to evaluation
outside the original domain $D$.

If we are given a function $u : D \to \F_q$ which is supposed to be the evaluations
of a linear form $\ell$ on $D$, we can ask about what the evaluation of this linear form
at a point $z \in S$ is. This is the viewpoint from which the DEEP lemma generalizes to general codes.

We start with two functions $u, u^* : D \to \F_q$ (which are supposed to correspond to linear forms, say $\ell \in \F_q^k$ and $\ell^* \in \F_q^k$.
We have a verifier who samples $z \in S$ and asking for $a = \ell(z)$ and $a^* = \ell^*(z)$. Given these answers, the verifier now samples $x\in \F_q$ and computes
$b=a^*+xa$ which is  supposedly equal to $\ell^*(z)+\ell(z)$ (if $u^*$ and $u$ are indeed codewords of $V$).
The result below says that if $S$ is the set of columns of an error correcting code with good distance, and $V$ has small list size for list-decoding up to radius $\slack$, then with high probability, the function $u_x=u^* + x u$
has distance at least $\approx \min\{\Delta(u^*, V), \slack\}$ from the sub-code of $V$ corresponding to the linear forms that evaluate to $b$ on $z$.

\newcommand{\inp}[1]{\langle#1\rangle}

\eli{remove robust}
\swastik{Actually now I think robust is good }
\begin{definition}
	[Robust] \label{def:robust sets}
	A set $S\subseteq \F^k$ is called $\sigma$-robust if every subset of $S$ of size $\sigma$ contains a basis for $\F^k$.
\end{definition}

The following claim is well-known in coding theory (cf. \cite[Problem 2.8]{RothBook}).

\begin{claim}\label{clm:code distance and sigma}
	Fix a full-rank matrix $G\in \F_q^{k\times N}, N\geq k$, and let
	$C=\condset{x \cdot M}{x\in\F_q^k}$ be the linear code generated by it. Then the set of columnss of $G$ is $\sigma$-robust  if and only if  the minimum distance of $C$ is at least $N -\sigma + 1$.
\end{claim}

%
%


\begin{lemma}[DEEP method for general linear codes]
	\label{lem:main general codes}
	Let $V$ be an $[n,k,d]_q$-code that is $(\slack,L^*_\slack)$-list decodable for some $\slack>0$, and fix $G\in \F_q^{k\times n}$ to be its generating matrix.
	Let $S\subset\F_q^k$
	be a $\sigma$-robust set of size $N$.
	For $z\in S,b \in \F_q$, let
	$$V_{z,b} = \{ v \in V \mid v=G\cdot \ell_v \text{ AND } \inp{\ell_v,z}=b \}$$
	where $\inp{v,z}=\sum_{i=1}^k v_i,z_i$.

	Let $u,u^* \in \F_q^n$.
	For each $z\in S$, let $B_z(X) \in \F_q[X]$ be an arbitrary linear function.
	Suppose that for some $\epsilon>0$ the following holds,	\begin{equation}\label{eq1}
	\Pr_{x\in \F_q,z\in S} [ \Delta( u^* + x u, V_{z, B_z(x)}) < \slack ] \geq
	\max\left(2L^*_\slack \left( \frac{\sigma}{N} + \epsilon\right)^{1/3} , \frac{4}{\epsilon^2 q} \right),
	\end{equation}
	Then there exist $v, v^*\in V$ and $C \subset [n]$ such that:
	\begin{itemize}
		\item $|C| \geq (1-\slack-\epsilon) n$,
		\item $u|_{C}=v|_{C}$,
		\item $u^*|_{C}=v^*|_{C}$.
	\end{itemize}
	Consequently, we have $\Delta(u, V), \Delta(u^*, V) \leq \slack + \epsilon$.
\end{lemma}

The proof is analogous to the proof in the Reed-Solomon case, and appears in~\cref{sec:DEEPgeneralproof}.

\swastik{I think the instantiations are hard to get excited about ... because this space $V_{z,b}$ has no meaning
for the reader (and also the headline results, DEEP-FRI and DEEP-ALI have no analogue with the instantiations).
In the next paper where we do DEEP-FFRI, the FRS instantiation will be much more motivated and exciting. So I removed the instantiations
for now -- they are moved to the end of the tex file.}

\paragraph{Discussion} For the special case of RS codes, the DEEP method can be used to locally modify the problem and reduce degree. Indeed, the subcode $V_{z,b}$ in the case of RS codes corresponds is comprised of functions $f:D\to \F$ that are evaluations of polynomials of degree $d$ whose interpolating polynomial $P_f$ satisfies $P_f(z)=b$. From such a codeword, one can construct a new codeword $f_{z,b}:D\to\F$ defined by $f_{z,b}(x)=\frac{f(x)-b}{z}$, which is well-defined for all $z\not \in D$. Notice that the transformation from $f$ to $f_{z,b}$ is \emph{1-local}, meaning that each entry of $f_{z,b}$ is constructed by making a single query to $f$. Furthermore, this transformation maps a subset of the code $RS[\F,D,d]$ to the code
$\RS[\F,D,d-1]$, so we may use this transformation in RS IOPPs (as will done in the following section).

In contrast, for a general $k$-dimensional linear code $V$, the subcode $V_{z,b}$, while being an affine subspace of $V$, has less structure. In particular, it is not clear how to locally convert this subcode to a ``nice'' code of dimension $k-1$. An interesting middle ground, left to future work, is the case of algebraic codes like Reed Muller codes and Algebraic Geometry codes which resemble RS codes.

\section{\DFRI}\label{sec:deep fri}

In this section we describe the new fast RS IOPP, called \DFRI.
We start by recalling the \FRI protocol from~\cite{FRI}, describing it nearly verbatim as in~\cite[Section 7]{Ben-SassonKS18_ECCC}.

\subsection{\FRI}\label{sec:FRI description}

Our starting point is a function $f\zr:L\zr\to\F$ where $\F$ is a finite field, the evaluation domain $L\zr\subset\F$ is a coset of a group\footnote{The group can be additive, in which case $\F$ is a binary field, or multiplicative, in which case it is not.}
contained in $\F$, and $|L\zr|=2^{k\zr}$.
We assume the target rate is $\Rate=2^{-\RateInt}$ for some positive integer $\RateInt$.
The \FRI protocol is a two-phase protocol (the two phases are called \commit and \query) that convinces a verifier that $f\zr$ is close to the Reed-Solomon code $\RS[\F, L\zr,  \Rate]$.

The \commit phase of the \FRI protocol
involves $\rounds=k\zr-\RateInt$ rounds.
Before any communication, the prover and verifier agree
on a sequence of (cosets of) sub-groups $L\ii$,
where $|L\ii| = 2^{k\zr - i}$. Let $\RS\ii$
denote the Reed-Solomon code $\RS[\F, L\ii, \Rate|L\ii|]$.

The main ingredient of the \FRI protocol
is a special algebraic hash function $H_x$,
which takes a seed $x \in \F$, and
given as input a function $f: L\ii \to \F$,
it produces as output a hash whose length is $1/2$ as long as $f$.
More concretely, $H_x[f]$ is a function
$$ H_x[f]: L\iip \to \F$$
with the following properties:
\begin{enumerate}
	\item{\bf locality:} For any $s \in L\iip$, $H_x[f](s)$ can be computed by querying $f$ at just two points in its domain (these two points are $(q\ii)^{-1}(s)$).
	\item{\bf completeness:} If $f \in \RS\ii$, then for all $x \in \F$, we have that $H_x[f] \in \RS\iip$.
	\item{\bf soundness:} If $f$ is far from $\RS\ii$, then
	%
	with high probability over the choice of seed $x$,
	$H_x[f]$ is quite far from $\RS\iip$.
\end{enumerate}
These last two properties roughly show that for random $x$,
$H_x$ preserves distance to Reed-Solomon codes.
For the precise description of $H_x$ see~\cref{sec:alghash} and \cite{Ben-SassonKS18}.

The high-level idea of the \FRI protocol can then be described as follows.
First we are in the \commit phase of the protocol.
The verifier picks a random $x\zr \in \F$ and asks the prover
to write down the hash $H_{x\zr}[f\zr]: L^{(1)} \to \F$.
By Properties 2 and 3 above, our original
problem of estimating the distance of $f\zr$ to $\RS\zr$ reduces
to estimating the distance of $H_{x\zr}[f\zr]$ to $\RS^{(1)}$ (which
is a problem of $1/2$ the size).
This process is then repeated: the verifier picks a random $x^{(1)} \in \F$
and asks the prover to write down $H_{x^{(1)}}[H_{x\zr}[f\zr]]$, and so on.
After $\rounds$ rounds of this, we are reduced to a constant sized problem which
can be solved in a trivial manner. However, the verifier cannot blindly trust
that the functions $f^{(1)}, \ldots $ that were written down by the prover truly
are obtained by repeatedly hashing $f\zr$. This has to be checked, and the verifier
does this in the \query phase of the protocol, using Property 1 above.

We describe the phases of the protocol below.

\medskip
\noindent{\bf \commit Phase:}
\begin{enumerate}
	\item For $i = 0$ to $\rounds-1$:
	\begin{enumerate}
		\item The verifier picks uniformly random $x\ii \in \F$ and sends it to the prover.
		\item The prover writes down a function $f\iip: L\iip \to \F$.
(In the case of an honest prover,
		$f\iip = H_{x\ii}[f\ii]$.)

	\end{enumerate}
	\item The prover writes down a value $C \in \F_q$. (In the case of an honest
prover, $f\fin$ is the constant function with value $=C$).
\end{enumerate}

\medskip
\noindent{\bf \query Phase:} (executed by the Verifier)
\begin{enumerate}
	\item Repeat $\ell$ times:
	\begin{enumerate}
		\item Pick $s\zr \in L\zr$ uniformly at random.
		\item For $i = 0$ to $\rounds-1$:
		\begin{enumerate}
			\item Define $s\iip \in L\iip$ by $s\iip = q\ii(s\ii)$.
			\item Compute $H_{x\ii}[f\ii](s\iip)$ by making 2 queries to $f\ii$.
			\item If $f\iip(s\iip) \neq H_{x\ii}[f\ii](s\iip)$, then REJECT.
		\end{enumerate}
		\item If $f\fin(s\fin) \neq  C$, then REJECT.
	\end{enumerate}
	\item ACCEPT
\end{enumerate}
\medskip


The previous state of the art regarding the soundness of \FRI is given by the following statement from \cite{Ben-SassonKS18}.
In what follows let $J_\epsilon(x)=1-\sqrt{1-x(1-\epsilon)}$.

\begin{theorem}
	[\FRI soundness (informal)]
	\label{thm:FRI BKS18}
		Suppose $\delta\zr\eqndef\Delta({f\zr,\RS\zr})>0$. Let $n = |L\zr|$. Then for any $\epsilon>0$ there exists $\epsilon'>0$ so that with probability at least
		\begin{equation}
		\label{eq:soundness commit formula old}
		1-\frac{2 \log n}{\epsilon^3|\F|}
		\end{equation}
		over the randomness of the verifier during the \commit phase,
		and for any (adaptively chosen) prover oracles $f\one,\ldots, f\fin$,
		the \query protocol with repetition parameter $\RepInt$ outputs $\acc$ with probability at most
		\begin{equation}\label{eq:soundness query formula old}
		\left(1-\min\set{\delta\zr, 1-(\rho^{1/4}+\epsilon')}+\epsilon \log n  \right)^\RepInt.
		\end{equation}

\end{theorem}

\begin{remark}\label{rem:improved fri analysis}
Using the improved distance preservation of \cref{lem:one-and-a-half-positive} in
the analysis of \FRI from~\cite{Ben-SassonKS18}, one immediately improves the factor $1/4$ in the exponent in \cref{eq:soundness query formula new} to an exponent of $1/3$ (details omitted).
\end{remark}

\subsection{\DFRI}
We now describe our variation of \FRI,
that we call \DFRI, for which we can give improved soundness guarantees,
at the cost of a small increase in the query complexity (but no increase
in the proof length or the number of queries to committed proofs -- which is important in
applications).

Before we can describe our protocol we introduce the
operation of ``quotienting", which allows us to
focus our attention on polynomials taking certain values
at certain points.

\subsubsection{Quotienting} \label{sec:quotienting}

Suppose we a set $L \subseteq \F_q$
and a function $f : L \to \F_q$.
Suppose further that we are given a point $z \in \F_q$
and a value $b \in \F_q$.

We define the function $\QUOTIENT(f, z, b): L \to \F_q$ as
follows. Let $Z(Y) \in \F_q[Y]$ be the polynomial $Z(Y)= Y-z$.
Then we define $\QUOTIENT(f, z, b)$ to be the function $g: L \to \F_q$ given by:
$$ g(y) = \frac{f(y) - b}{Z(y)}$$
(or more succinctly, $ g = \frac{f-b)}{Z}$).

\begin{lemma}\label{lem:quotient}
Let $L \subseteq \F_q$. Let $z \in \F_q$ with $z \not\in L$.
Let $d \geq 1$ be an integer.

Let $f: L \to \F_q$, and $b \in \F_q$.
Let $g = \QUOTIENT(f,z,b)$.
Then the following are equivalent:
\begin{itemize}
\item There exists a polynomial $Q(X) \in \F_q[X]$
of degree at most $d-1$ such that $\Delta(g, Q) < \delta$.
\item There exists a polynomial $R(X) \in \F_q[X]$
of degree at most $d$ such that $\Delta(f, R) < \delta$
and $R(z) = b$.
\end{itemize}
\end{lemma}
\begin{proof}
If there is such a polynomial $Q, \deg(Q)\leq d-1$ that agrees with $g$ on all but a $\delta$-fraction of entries, we can take
$R = QZ + b$. Notice $\deg(R)\leq d$ because $\deg(Z)=1$.

Conversely, if there is such a polynomial $R$ that agrees with $f$ on all but a $\delta$-fraction of entries,
we can take  $Q = (R - b)/Z$. This is indeed a polynomial
because $R-b$ vanishes on $z$, so $Z|(R-b)$ in the ring of polynomials.

Finally, by construction $R$ agrees with $f$ whenever $g$ agrees with $R$ and this completes the proof.
\end{proof}

\subsection{\DFRI}

Recall: We have linear spaces $L\zr, L\one, \ldots, L\fin$,
with dimensions $k, k-1, \ldots, k-\rounds$.
We further have 1 dimensional subspaces $L_0\zr, L_0\one, \ldots, L_0\fin$ with $L_0\ii \subseteq L\ii$.

For this, it will be helpful to keep in mind the case that the domain $L\zr$
is much smaller than the field $\F_q$ (maybe $q = |L\zr|^{\Theta(1)}$).

\begin{protocol}[\DFRI]~ \\
Input: a function $f\zr : L\zr \to \F_q$ which is supposed to be of degree $< d\zr$.

\medskip
\noindent{\bf \commit Phase:}
\begin{enumerate}
\item For each $i \in [0, \rounds - 1]$:
\begin{enumerate}
\item The verifier picks a uniformly random $z\ii \in \F_q$.
\item The prover writes down a degree one polynomial $B\ii_{z\ii}(X)\in\F_q[X]$
(which is supposed to be such that $B\ii_{z\ii}(x)$ equals the evaluation
of the low degree polynomial $H_x[f\ii]$ at $z\ii$).
\item The verifier picks uniformly random $x\ii \in \F_q$.
\item The prover writes down a function
$$f\iip : L\iip \to \F_q.$$
(which on input $y$ is supposed to equal $\QUOTIENT( H_{x\ii}[f\ii], z\ii, B\ii_{z\ii}(x))$.)
\end{enumerate}
\item The prover writes down a value $C \in \F_q$.
\end{enumerate}

\medskip
\noindent{\bf \query Phase:} ~
\begin{enumerate}
\item Repeat $\RepInt$ times:
\begin{enumerate}
\item The verifier picks a uniformly random $s\zr \in D$.
\item For each $i \in [0, \rounds - 1]$:
\begin{enumerate}
\item Define $s\iip \in L\iip$ by $s\iip = q\ii(s\ii)$.
\item Compute $H_{x\ii}[f\ii](s\iip)$ by making 2 queries to $f\ii$.
\item If $ H_{x\ii}[f\ii](s\iip) \neq f\iip(s\iip) \cdot ( s\iip- z\ii) + B\ii_{z\ii}(x\ii)$, then REJECT.
\end{enumerate}
\item If $f\fin(s\fin) \neq C$, then REJECT.
\end{enumerate}
\item ACCEPT.
\end{enumerate}
\end{protocol}

\subsection{Analysis}

The following theorem proves the soundness of the \DFRI protocol.

\begin{theorem}
	[\DFRI]
	\label{thm:DFRI soundness}
	Fix degree bound $d\zr = 3 \cdot 2^{\rounds}-2$ and $\RS\zr = \RS[\F_q, L\zr, d\zr]$. Let $n = |L\zr|$.

	For some $\epsilon, \delta > 0$, let
$$\deltastar = \delta - 2\rounds\epsilon,$$
$$\Lstar = \calL(\F_q, L\zr, d\zr, \deltastar),$$
$$\nustar =  2\Lstar\left(\frac{d\zr}{q}+\epsilon\right)^{1/3}+\frac{4}{\epsilon^2 q}.$$

	Then
	the following properties hold when the \DFRI protocol
	is invoked on oracle $f\zr:L\zr\to\F_q$,
	\begin{enumerate}
		\item{\bf Prover complexity}\label{itm:thm:RS scalability}
		is $O(n)$ arithmetic operations over $\F$
		\item{\bf Verifier complexity}\label{itm:thm:RS succintness}
		is $O(\log n)$ arithmetic operations over $\F$ for a single invocation of the \query phase; this also bounds communication and query complexity (measured in field elements).
		\item{\bf Completeness}
		\label{itm:completeness}
		If $f\zr\in\RS\zr$
		and $f\one,\ldots, f\fin$ are computed by the prover specified in the \commit
		phase, then the \DFRI verifier outputs $\acc$ with probability $1$.


		\item{\bf Soundness}
		\label{itm:soundness new}
		Suppose $\Delta(f\zr,\RS\zr) > \delta$.
		Then with all but probability
		\begin{equation}
		\label{eq:soundness commit formula new}
		\error_\commit \leq \rounds \cdot \nustar \leq (\log n) \cdot \nustar.
		\end{equation}
		and for any (adaptively chosen) prover oracles $f\one,\ldots, f\fin$,
		the \query protocol with repetition parameter $\RepInt$ outputs $\acc$ with probability at most
		\begin{equation}\label{eq:soundness query formula new}
		\error_\query \leq \left(1- \deltastar + (\log n) \cdot \epsilon \right)^\RepInt
		\end{equation}
\swastik{DO WE NEED THE FOLLOWING?}
		Consequently, the soundness error of \FRI is at most
		\begin{equation}\label{eq:soundness formula new}
		\error\left(\delta\right) \leq
		(\log n) \cdot \nustar +\left(1- \deltastar + (\log n) \cdot \epsilon  \right)^\RepInt
		\end{equation}
	\end{enumerate}
\end{theorem}

We give a consequence below with a specific setting of parameters
based on the Johnson bound.
\begin{example} \label{cor:deep-fri-johnson}
	Continuing with the notation of \cref{thm:DFRI soundness},
	fix degree bound $d\zr = 3 \cdot 2^{\rounds}-2$ and assume
	$n = |L\zr| < \sqrt{q}$.
	Let $\RS\zr = \RS[\F_q, L\zr, d\zr]$ and
	let $\rho = d\zr/n$ be its rate.

	Let $f\zr : L\zr \to \F_q$ be a function,
and let $\delta\zr = \Delta(f\zr , RS\zr) $.
	Then with all but probability $\error_\commit \leq O(q^{-\Omega(1)})$, the query phase will accept with probability at most:
	$\error_\query \leq (\max(1-\delta\zr, \sqrt{\rho}) + o(1)) ^{\RepInt}$
	as $n \to \infty$.
\end{example}
\begin{proof}
	Note that $d\zr \leq n \leq \sqrt{q}$.

	Set $\delta = \min(\delta\zr, 1 - \sqrt{\rho} - q^{-1/13})$, and apply the previous theorem.
	\cref{thm:johnson} implies that $\Lstar < q^{1/13}/(2\sqrt\rho) = O(q^{1/13})$.
	Set $\epsilon = q^{-6/13}$.
	Hence \[\nustar = 2\Lstar\left(d\zr q^{-1}+q^{-6/13}\right)^{1/3}+4q^{-6/13}
	= O(q^{-1/13}),\]
	which implies $\error_\commit \leq \tilde{O}(q^{-1/13})$.

	If $\delta = \delta\zr$, then $1 - \deltastar + (\log n)\epsilon = 1-\delta + o(1)$.
	Otherwise $\delta = 1 - \sqrt{\rho} - q^{-1/13}$, 
and so  $$1-\deltastar + (\log n) \epsilon =
	\sqrt\rho + q^{-1/13} + (\log n) \epsilon =
	\sqrt\rho + q^{-1/13} + (\log n) q^{-6/13}.$$
	Thus $\error_\query \leq (\max(1-\delta, \sqrt{\rho})+o(1))^{\RepInt}$.
\end{proof}

We now give an example setting of \DFRI under the optimistic
\cref{conj:L}.

\begin{example} \label{cor:deep-fri-conjL}
	Assume Conjecture~\ref{conj:L}.
	Continuing with the notation of \cref{thm:DFRI soundness},
	fix degree bound $d\zr = 3 \cdot 2^{\rounds}-2$ and $n = |L\zr|$.
	Let $\RS\zr = \RS[\F_q, L\zr, d\zr]$ and
	let $\rho = d\zr/n$ be its rate.

	Let $C = C_\rho$ be the constant given by Conjecture~\ref{conj:L}.
	Suppose $q > n^{24C}$.

	Let $f\zr : L\zr \to \F_q$ be a function,
and let $\delta\zr = \Delta(f\zr , RS\zr) $.
	Then with all but probability $\error_\commit \leq O(q^{-\Omega(1)})$, the query phase will accept with probability at most:
	$\error_\query \leq (1-\delta\zr + o(1)) ^{\RepInt}$
	as $n \to \infty$.
\end{example}
\begin{proof}
Set $\epsilon = q^{-1/(6C)}$.

Set $\delta = \min(\delta\zr, 1 - \rho - q^{-1/(6C)})$.
Conjecture~\ref{conj:L} gives us:
$$ \Lstar <  n^C q^{1/6}.$$

We then apply the previous theorem.
We get $\nustar \ll O(\Lstar \cdot (d/q + \epsilon)^{1/3} + \frac{1}{\epsilon^2 q}) \ll q^{-1/12}$,
and this gives us the claimed bound on $\error_\commit$.

For the bound on $\error_\query$, we note that $\delta = \delta\zr + o(1)$.
This is because {\em every} function is within distance $1-\rho$ of $\RS\zr$ (this follows easily from
polynomial interpolation). Thus 
$$1-\deltastar + (\log n) \epsilon = \delta + o(1),$$
and we get the desired bound on $\error_\commit$.
\end{proof}

The prover and verifier complexity as well as completeness follow by
construction (see, e.g., \cite{FRI} for detailed analysis of these aspects). In the rest of the section we prove the soundness bound
of \cref{thm:DFRI soundness}.

\subsection{Preparations}

We do the analysis below for the case $\RepInt = 1$. The generalization to arbitrary $\RepInt$ easily
follows.

Define $d\zr = 3\cdot 2^{\rounds} - 2$, and $d\iip = d\ii/2 - 1$. It is easy to check that $d\fin = 1$.
Define $\RS\ii = \RS[\F_q, L\ii, d\ii]$. In the case of the honest prover (when $f\zr \in \RS\zr$),
we will have $f\ii \in \RS\ii$ for all $i$.

Our analysis of the above protocol will track the agreement of $f\ii$ with $\RS\ii$.
This agreement will be measured in a certain weighted way, which we define next.

\subsubsection{The success probability at $s \in L\ii$}

There is a natural directed forest that one can draw on the vertex set
$$ L\zr \cup L\one \cup \ldots \cup L\fin,$$
namely, where $s \in L\ii$ is joined to $q\ii(s) \in L\iip$
(and we say that $s$ is a child of $q\ii(s)$).
Note that every vertex not in $L\zr$ has two children.

Let $i \leq \rounds-1$ and $s_0 \in L\ii$. Let $s\in L\iip$ be the parent
of $s_0$, and let $s_1 \in L\ii$ the sibling of $s_0$.
We color $s_0$ GREEN if $f\iip(s)$ is consistent with $f\ii\mid_{\{s_0, s_1\}}$
according to the test
$$ H_{x\ii}[f\ii](s) =  f\iip(s) \cdot ( s- z\ii) + B\ii_{z\ii}(x\ii)$$
and we color $s_0$ RED otherwise. Notice that a vertex and its sibling get the same color.

For $s \in L\fin$, we color $s$ GREEN if $f\fin(s) = C$
and RED otherwise.

The \query phase of the protocol can be summarized as follows:
we pick a uniformly random $s\zr \in L\zr$ and consider the path
$s\zr, s\one, s\two \ldots, s\fin$ going through all its ancestors. If all
these vertices are GREEN, then we ACCEPT, otherwise we REJECT.

To capture this, we define functions $\eta\ii: L\ii \to \mathbb R$ as follows.
For $s \in L\ii$, let $\eta\ii(s)$ be the fraction of leaf-descendants $s\zr$ of $s$
for which the path from $s\zr$ to $s$ (including $s\zr$ but not including $s$)
consists exclusively of GREEN vertices.
Observe that $p_{ACCEPT} = \mathbb E_{s \in L\fin}[\eta\fin(s) \cdot 1_{f\fin(s) = C}]$ equals the probability that the
\query phase accepts.

The exact quantity that we will track is as $i$ increases is the weighted agreement:
$$\alpha\ii= \agree_{\eta\ii}[f\ii, \RS\ii].$$
Notice that
$$\alpha\zr = 1 - \Delta(f\zr, \RS\zr),$$
and the acceptance probability, $p_{ACCEPT}$ satisfies:
$$p_{ACCEPT} \leq \alpha\fin.$$

Our main intermediate claim is that with high probability over the choice of
$x\ii,  z\ii, B\ii_{z\ii}$, we have that $\alpha\iip$ is
not much more than $\alpha\ii$. This gives us that $p_{ACCEPT}$
is not much more than $1 - \Delta(f\zr, \RS\zr)$, as desired.

\subsubsection{Operations $\AVG$ and $\ZERO$}

We define two important operations.
\begin{enumerate}
\item $\AVG$. For a function $w: L^{(i-1)} \to \mathbb R$, we define
the function $\AVG[w] : L^{(i)} \to \mathbb R$ as follows.
Let $s \in L^{(i)}$,
and let $\{ s_0, s_1 \} = (q^{(i-1)})^{-1}(s)$.
Then define:
$$\AVG[w](s) = \frac{w(s_0) + w(s_1)}{2}.$$
\item $\ZERO$. For a function $w: L^{(i)} \to \mathbb R$, and a set
$S \subseteq L^{(i)}$, we define
the function $\ZERO[w,S]: L^{(i)} \to \mathbb R$ as follows.
For $s \in L^{(i)}$, we set:
$$ \ZERO[w,S](s) = \begin{cases} 0 & s \in S \\ w(s) & s \not\in S \end{cases}.$$
\end{enumerate}

We can use these two operations to express $\eta\iip$ in terms of $\eta\ii$.
Let $E\iip$ denote the set of all $s \in S\iip$ both of whose
children are RED (i.e., the test
$$ H_{x\ii}[f\ii](s) =  f\iip(s) \cdot ( s- z\ii) + B\ii_{z\ii}(x\ii)$$
fails at $s$).

Define $\theta\iip : L\iip \to \mathbb R$ by
$$ \theta \iip = \AVG[\eta\ii].$$
Then we have:
$$ \eta\iip = \ZERO(\theta\iip, E\iip).$$
Analogous to our definition of
$$\alpha\ii = \agree_{\eta\ii}(f\ii, \RS\ii),$$
we define
$$\beta\iip = \agree_{\theta\iip}(H_{x\ii}[f\ii], \{P(Y) \in \F_q[Y] \mid \deg(P) \leq d\iip \mbox{ and } P(z\ii) = B\ii_{z\ii}(x\ii)\} ).$$


The following two lemmas control the growth of $\alpha\ii$ and $\beta\ii$.

\begin{lemma}
\label{alphabeta1}
For all $i$, with probability at least $ 1- \nustar$ over the choice of $x\ii, z\ii$,
we have:
$$ \beta\iip \leq \max(\alpha\ii, 1-\deltastar)  + \epsilon.$$
\end{lemma}
We prove this using \cref{lem:weighted-OOD-main} in Appendix~\ref{sec:alphabetalemmas}.
\begin{lemma}
\label{alphabeta2}
For all $i$,
$$\alpha\ii \leq \beta\ii.$$
\end{lemma}
We prove this using~\cref{lem:quotient} in Appendix~\ref{sec:alphabetalemmas}.

We can now complete the proof of~\cref{thm:DFRI soundness}.

\begin{proof}
As observed earlier, $\alpha\zr = 1 - \Delta(f\zr, \RS\zr) < 1 - \delta$.

The two lemmas above imply that with probability at least $1 - \rounds \nustar$,
$$\alpha\fin \leq \max(\alpha\zr, 1-\deltastar) + \rounds \cdot \epsilon < (1 - \min(\delta, \deltastar) + \rounds \cdot \epsilon).$$

Finally, we use the observation that $p_{ACCEPT} \leq \alpha\fin$ to complete the proof.
\end{proof}

\section{The DEEP Algebraic Linking IOP (\DEEPSTIK{}) protocol}
\label{sec:Deep stik}

\def\RAPR{R_{\mathrm{APR}}}
\def\RAIR{R_{\mathrm{AIR}}}
\def\Qlcm{Q_{\mathrm{lcm}}}
\def\xx{\mathbbmss x}
\def\ww{\mathbbmss w}
\def\maxdeg{d_{\mathcal C}}
\def\fpoly{\tilde f}
\def\gpoly{\tilde g}
\def\foracle{f}
\def\goracle{g}
\def\Tarith{T_{\text{arith}}}
\def\width{{\sf{w}}}
\def\nqueries{{\sf{Q}}}

The techniques used earlier in \cref{lem:main,sec:deep fri} can
also be used to improve soundness in other parts of an interactive
oracle proof (IOP) protocol. We apply them here to obtain a Scalable Transparent IOP of Knowledge (STIK) \cite[Definition 3.3]{stark}
with better soundness than the prior state of the art, given
in \cite[Theorem 3.4]{stark}.

Proof systems typically use a few steps of reduction to convert problems
of membership in a nondeterministic language $L$ to algebraic problems
regarding proximity of a function (or a sequence of functions) to
an algebraic code like Reed-Solomon (or, in earlier works, Reed-Muller).
The goal of such a reduction is to maintain a large \emph{proximity gap} $\gamma$, meaning that for instances in $L$, an honest prover will provide information that leads to codewords,
whereas for instances not in $L$, any oracles submitted by the prover will be converted by the reduction, with high probability, to functions that are $\gamma$-far from the code. Considerable effort is devoted to increasing $\gamma$ because
it is the input to the proximity protocols (like \FRI and \DFRI) and the soundness of those protocols is correlated to $\gamma$ (as discussed earlier, e.g., in \cref{thm:DFRI soundness}).

The STIK protocol is a special case of this paradigm. It requires the prover to provide oracle access to a function $f:D\to\F$ that is supposedly an RS encoding of a witness for membership of the input instance in $L$. A set of $t$-local constraints is applied to $f$ to construct a function $g:D\to\F$, along with a gap-gurantee: If $f$ is indeed an encoding of a valid witness for the instance, then the resulting function $g:D\to\F$ is also be a member of an RS code.
One of the tests that the verifier performs is a \emph{consistency test} between $f$ and $g$, and, prior to this work, this consistency test was applied to the functions $f$ and $g$ \emph{directly}. This leads to a rather small gap $\gamma\leq \frac{1}{8}$ which results in a small soundness guarantee from the RPT protocol applied to $f,g$ later on.

In this section we apply the DEEP technique to this setting. After $f$ and $g$ have been provided, the verifier samples a random $z \in \F_q$ and asks
for the values of the interpolating polynomials of $f,g$ on all
$t$ entries needed to check the consistency test. Our verifier now
applies the $\QUOTIENT$ operation to $f,g$, using the information
obtained  from the prover. Crucially, we prove that a \emph{single} consistency test, conducted over a large domain $D'\supset D$, suffices to improve the proximity gap to roughly $1-\sqrt{\rho}$, a value
that approaches $1$ as $\rho\rightarrow 0$. Assuming \cref{conj:L} the proximity gap is nearly-optimal, at $\gamma\approx 1-\rho$ (compare with with the value $\gamma\leq 1/8$ obtained by prior works). Details follow.

We focus on the the Algebraic linking IOP protocol (\ALI) of  \cite[Theorem B.15]{stark},
and present
a new protocol that we call \DEEPSTIK (\cref{protocol:deep-stik}) that
obtains the aforementioned improved proximity gap(s).

In what follows, we will first recall (a variant of)
the language (or, more accurately, binary relation) which was the input to the \ALI protocol of \cite{stark} and is likewise the input to our \DEEPSTIK{}
protocol. The description of the protocol follows in \cref{sec:deep stik protocol specification}. Its basic properties are specified in \cref{sec:deep stik properties} and we analyze its soundness in
\cref{thm:deep-stik,sec:deep-stik-soundness}.

\subsection{The Algebraic Placement and Routing (APR) Relation}\label{sec:APR definition}

In what follows we use the notation $\fpoly$ to refer to a polynomial in $\F[x]$.
Note that the operator $\mid_D$ for $D\subseteq \F$
takes a polynomial to a function: $\fpoly\mid_D: D \to \F$.

We start by defining a simplified version of the Algebraic placement and routing
relation (APR). See \cite[Definition B.10]{stark}.
In particular, we only use one witness polynomial.
This relation will be the input to the reduction used in \cref{protocol:deep-stik}.

\begin{definition} \label{def:apr}
The relation $\RAPR$ is the set of pairs $(\xx, \ww)$ satisfying:
\begin{enumerate}
\item \textbf{Instance format:}
The instance $\xx$ is a tuple $(\F_q, d, \mathcal C)$
where:
\begin{itemize}
\item $\F_q$ is a finite field of size $q$.
\item $d$ is an integer representing a bound on the degree of the witness.
\item $\mathcal C$ is a set of $|\mathcal C|$ tuples $(M^i, P^i, Q^i)$ representing constraints.
$M^i$ is the \emph{mask} which is a sequence of
field elements $M^i = \{ M^i_j \in \F_q\}_{j=1}^{|M^i|}$.
$P^i$ is the \emph{condition} of the constraint which is a polynomial with $|M^i|$ variables.
$Q^i \in \F_q[x]$ is the \emph{domain polynomial}
of the constraint which should vanish on the locations where
the constraint should hold.
\end{itemize}

We further introduce the following notation:
\begin{itemize}
\item Let $\mathcal M = \{ M_j^i  \mid 1 \leq i \leq |\mathcal C|$ and $1 \leq j \leq |M^i| \} \subseteq
\F_q$ be the \emph{full mask}.
\item Let $\maxdeg = \max_{i} \deg(P^i)$
be the maximal total degree of the $P^i$s.
\item Let $\Qlcm \in \F_q[x]$ be the least common multiple of the $Q^i$s.
\end{itemize}

\item \textbf{Witness format:}
The witness $\ww$ is a polynomial $\fpoly \in \F_q[x]$.
A constraint $(M, P, Q)$ is said to hold at a location $x \in \F_q$ if
$P(\fpoly(x \cdot M_1), \ldots, \fpoly(x \cdot M_{|M|})) = 0$.
We say that $\fpoly$ satisfies the constraint
if the constraint holds at every $x \in \F_q$ for which $Q(x) = 0$.

We say that $\ww$ satisfies the instance if and only if
$\deg(\fpoly)<d$ and $\fpoly$ satisfies all of the constraints.
\end{enumerate}
\end{definition}

\def\gen{\gamma}
To see that the notion of the $\RAPR$ relation defined above
is strong enough\lior{Check},
we follow the ideas from \cite{stark} and show a reduction
from an Algebraic Intermediate Representation (AIR, see \cite[Definition B.3]{stark}) to an APR.
The following uses the notation from \cite[Definition B.3]{stark}.
Let $\xx = (\F_q, T,\width, \mathcal P, C, \textsf{B})$
be an instance of $\RAIR$\lior{What is $C$?}.
Pick a multiplicative subgroup
$\langle \gen \rangle \subseteq \F_q^\times$ of size $T \cdot \width$
and pick $\fpoly$ such that $\fpoly(\gen^{t\width + j}) = w_j(t)$
for $t \in [T]$ and $i \in [\width]$ (here $[n] = \{0, \ldots, n-1\}$).
For all the constraints in $\mathcal P$,
choose the mask $M = \{1, \gen, \ldots, \gen^{2\width-1}\}$ and choose
the domain polynomial whose zeros are $\{\gen^{t\width}\}_{t \in [T-1]}$
($Q(x) = (x^{T} - 1) / (x - \gen^{-\width})$).
Replace each boundary constraint $(i, j, \alpha) \in \textsf{B}$
with a regular constraint with mask $M = \{1\}$,
$P(x) = x - \alpha$ and $Q(x) = x - \gen^{i\width + j}$.

\subsection{The \DEEPSTIK protocol}\label{sec:deep stik protocol specification}

We now describe our new protocol, that achieves improved soundness, as stated in the following theorem.

\begin{theorem}[\DEEPSTIK soundness]
\label{thm:deep-stik}
Fix a code rate $0 < \rho < 1$ and a distance parameter $0 < \delta \leq 1-\rho$.
Let $D, D' \subseteq \F_q$ be two evaluation domains
such that $|D| = d\rho^{-1}$ and $|D'| = d \cdot \maxdeg \rho^{-1}$.
Let $\RPT_D$, $\RPT_{D'}$ be two IOPPs with perfect completeness
for the codes $\RS[\F_q, D, (d - |\mathcal M|)/|D|])$ and
$\RS[\F_q, D', (d\maxdeg - 1)/|D'|])$ respectively.
Let $\epsilon, \epsilon'$ be the bounds on the soundness error (acceptance probability)
for words that are at least $\delta$-far from the corresponding code.
Denote \[L = \max\{\calL(\F_q, D, d, \delta),
\calL(\F_q, D', d \cdot \maxdeg, \delta)\}.\]
Then, there exists an IOP for $\RAPR$ with perfect completeness
and soundness error $\epsilon + \epsilon' + \frac{2L^2 (d \cdot \maxdeg + \deg(\Qlcm))}{q}$.
\end{theorem}

\begin{example}
Fix a code rate $0 < \rho < 1$.
Choosing \DFRI as the $\RPT$ protocol and setting
$\delta = 1 - \sqrt\rho - q^{-1/13}$
as in \cref{cor:deep-fri-johnson},
using $\RepInt$ repetitions,
we obtain an IOP for $\RAPR$ with
perfect completeness
and soundness error that approaches $2\rho^{\RepInt/2}$
as $q \to \infty$
assuming the parameters $d$, $\maxdeg$, $\deg(\Qlcm)$ of the APR
are constant with respect to $q$.
\end{example}
\begin{proof}
\Cref{thm:johnson} implies that $L \leq q^{1/13}/(2\sqrt{\rho}) = O(q^{1/13})$.
Hence the expression $2L^2 (d \cdot \maxdeg + \deg(\Qlcm))/q$ approaches $0$
as $q \to \infty$.
Moreover, \cref{cor:deep-fri-johnson} implies that $\epsilon, \epsilon'$
approach $\rho^{\RepInt/2}$.
\end{proof}

We now describe the protocol that achieves the soundness of \cref{thm:deep-stik}.

\begin{protocol}[\DEEPSTIK{}] \label{protocol:deep-stik} ~
\begin{enumerate}
\item \label{itm:stark-step-first} The prover sends an oracle $\foracle: D \to \F$
(which should be $\fpoly\mid_D$).
\item The verifier sends random coefficients
$\alpha = (\alpha_1, \ldots, \alpha_{|\mathcal C|}) \in \F_q^{|\mathcal C|}$.
\item \label{itm:stark-step-last} The prover sends an oracle $\goracle_{\alpha}: D' \to F$
(which should be $\gpoly_{\alpha}\mid_{D'}$,
where
\begin{equation} \label{eqn:composition-poly}
\gpoly_{\alpha}(x) = \sum_{i=1}^{|\mathcal C|} \alpha_i \cdot
\frac{P^i(\fpoly(x \cdot M^i_1), \ldots, \fpoly(x \cdot M^i_{|M^i|}))}{Q^i(x)}.
\end{equation}
Note that $\deg(\gpoly_{\alpha}) < d \cdot \maxdeg$).
\item \label{itm:chooze-z}The verifier sends a random value $z \in \F_q$.
\item Denote $\mathcal M_z =
\{z \cdot M_j^i \mid 1 \leq i \leq |\mathcal C|$ and $1 \leq j \leq |M^i|\}$.
The prover sends $a_{\alpha, z}: \mathcal M_z \to \F$ (which should be
$\fpoly|_{\mathcal M_z}$).
The verifier deduces $b_{\alpha, z}$, the alleged value of $\gpoly_\alpha(z)$,
using \cref{eqn:composition-poly}.
\item Let
$U(x)$, $Z(x)$ as defined in \cref{sec:quotienting}
for $\QUOTIENT(\foracle, a_{\alpha, z})$ and let
\[ h^1(x) = h^1_{\alpha, z}(x) =
\QUOTIENT(\foracle, a_{\alpha, z}) =
\frac{\foracle(x) - U(x)}{Z(x)},\]
\[ h^2(x) = h^2_{\alpha, z}(x) =
\QUOTIENT(\goracle_{\alpha}, \{z \mapsto b_{\alpha, z}\}) =
\frac{\goracle_{\alpha}(x) - b_{\alpha,z}}{x - z},\]
\lior{Fix singleton notation}
and note that the verifier has oracle access to $h^1$ and $h^2$ using the oracles $\foracle$
and $\goracle_{\alpha}$.
\item They use $\RPT_D$ and $\RPT_{D'}$ to prove that $h^1$ is
at most $\delta$-far from $\RS[\F_q, D, (d - |\mathcal M|) / |D|]$
(in other words, it is close to a polynomial of degree $< d - |\mathcal M|$)
and that $h^2$ is
at most $\delta$-far from $\RS[\F_q, D', (d\maxdeg-1) / |D'|]$.
\end{enumerate}
\end{protocol}

\subsection{Properties of \DEEPSTIK}\label{sec:deep stik properties}

Note that in the original \ALI protocol, the equivalent to the expression
$P^i(\fpoly(x \cdot M^i_1), \ldots, \fpoly(x \cdot M^i_{|M^i|})) / Q^i(x)$
is sampled at $\nqueries$ random locations
from the evaluation domain, where $\nqueries$ is the number of queries.

The main idea in \DEEPSTIK is to use Quotienting to allow
the verifier to choose \emph{one} random element $z$ from the \emph{entire} field,
and check the consistency between $\fpoly$ and $\gpoly$ only at $x=z$.

The fact that \DEEPSTIK allows to sample from the entire field
introduces\lior{TODO} several advantages over the \ALI protocol from \cite{stark}:

\begin{description}
\item[Soundness] As described above, the reduction in \ALI has lower bound $1/8$
on the distance from the code for inputs that are not in the language, even for $\rho \to 0$.
In \DEEPSTIK{} the lower bound on the distance is $1 - \sqrt{\rho}$.
\item[Query complexity] In \ALI
the verifier queries
$|\mathcal M| \cdot \nqueries$
field elements
as we need $|\mathcal M|$ elements to evaluate
$P^i(\fpoly(x \cdot M^i_1), \ldots, \fpoly(x \cdot M^i_{|M^i|}))$.
In \DEEPSTIK{} the verifier queries $O(|\mathcal M| + \nqueries)$
field elements as the evaluation of $P^i$ is done once.
\item[Verifier complexity] Previously, the verifier complexity was
$\Omega(\nqueries \cdot \Tarith)$ (where $\Tarith$ is the arithmetic
complexity of evaluating all the constraints).
The verifier complexity in \DEEPSTIK{} depends
on $\nqueries + \Tarith$ as we evaluate the constraints only once.
\item[Prover complexity] It is possible to alter
\cref{def:apr} and \DEEPSTIK{} to work with several witness
polynomials $f_1, \ldots, f_{\width}$ (as was done in \ALI).
The prover complexity in this case will depend on
$(\width \rho^{-1} + \maxdeg \rho^{-1} + \Tarith \maxdeg)\cdot d$ instead of
$(\width \maxdeg \rho^{-1} + \Tarith \maxdeg) \cdot d$ (in \ALI).
\end{description}

\subsection{Soundness analysis} \label{sec:deep-stik-soundness}

The proof of \cref{thm:deep-stik} will follow from the following lemma:

\begin{lemma}
	Let $\event$ be the event that the \DEEPSTIK verifier accepts.
    If \[\Pr[\event] \geq \epsilon + \epsilon' + \frac{2L^2 (d \cdot \maxdeg + \deg(\Qlcm))}{q},\]
    then there exists a polynomial of degree $<d$ satisfying the constraints.
\end{lemma}
\begin{proof}
Let $L(\foracle) \subseteq \RS[\F_q, D, \rho]$ be the
set of codewords that are at most $\delta$-far from $\foracle$.
Similarly define $L(\goracle_{\alpha})$.
We have $|L(\foracle)|, |L(\goracle_{\alpha})| \leq L$.

Let $\event_1$ be the event where the verifier accepts and
$h^1$ and $h^2$ are at most
$\delta$-far from the corresponding codes.
Denote $\eta = 2L^2 (d \cdot \maxdeg + \deg(\Qlcm))/q$. Then, $\Pr[\event_1] \geq \eta$.
$\event_1$ implies that there exists a polynomial $\tilde h^1 = \tilde h^1_{\alpha, z}$
of degree $<d - |\mathcal M|$ such that
$\big|\{x \in D : \tilde h^1(x) \neq \frac{\foracle(x) - U(x)}{Z(x)}\}\big| < \delta |D|$.
Hence $Z(x)\cdot \tilde h^1(x) + U(x) \in L(\foracle)$. Similarly
there exists a polynomial $\tilde h^2 = \tilde h^2_{\alpha, z}$
of degree $<d\cdot \maxdeg - 1$ such that
$(x-z) \tilde h^2(x) + b \in L(\goracle_{\alpha})$.

Fix $\tilde r^1(x)$
(independent of $\alpha$ and $z$)
to be the element in $L(\foracle)$ maximizing the probability that
$\tilde r^1(x) = Z(x)\cdot \tilde h^1(x) + U(x)$ given $\event_1$.
Let $\event_2 \subseteq \event_1$ be the event that $\tilde r^1(x) = Z(x)\cdot \tilde h^1(x) + U(x)$.
It follows that $\Pr[\event_2] \geq \eta/L$.

Fix $\tilde r^2_{\alpha}(x) \in L(\goracle_{\alpha})$ maximizing the probability that
$\tilde r^2_{\alpha}(x) = (x-z) \tilde h^2(x) + b$ given $\event_2$
(note that $\tilde r^2_{\alpha}$ depends on $\alpha$
as the oracle $\goracle_{\alpha}$
was sent only after the verifier sent $\alpha$),
and let $\event_3 \subseteq \event_2$ be the event where
$\tilde r^2_{\alpha}(x) = (x-z) \tilde h^2(x) + b$.
We have $\Pr[\event_3] \geq \eta/L^2$.
This implies, $\Pr_{\alpha}[\Pr_z[\event_3] \geq \eta/(2L^2)] \geq \eta/(2L^2)$.

The event $\event_3$ implies
\begin{align*} \label{eqn:r1-good}
&\tilde r^1 \mid_{\mathcal M_z} = U\mid_{\mathcal M_z} = a_{\alpha, z},\\
&\tilde r^2_{\alpha}(z) = b_{\alpha, z}.
\end{align*}
Recall that $b_{\alpha, z}$ was defined according to \eqref{eqn:composition-poly}, so
\[
b_{\alpha, z} = \sum_{i=1}^{|\mathcal C|} \alpha_i \cdot
\frac{P^i(a_{\alpha, z}(z \cdot M^i_1), \ldots,
a_{\alpha, z}(z \cdot M^i_{|M^i|}))}{Q^i(z)}.
\]
Substituting values for $a_{\alpha, z}$ and $b_{\alpha, z}$
and multiplying by $\Qlcm(z)$ we obtain:
\begin{equation} \label{eqn:composition-poly-r}
\Qlcm(z) \cdot \tilde r^2_{\alpha}(z) = \sum_{i=1}^{|\mathcal C|} \alpha_i \cdot
P^i(\tilde r^1(z \cdot M^i_1), \ldots,
\tilde r^1(z \cdot M^i_{|M^i|}))\cdot \frac{\Qlcm(z)}{Q^i(z)}.
\end{equation}

Both sides of the equation are polynomials of degree $<d\cdot \maxdeg + \deg(\Qlcm)$
in $z$.
For every $\alpha$ for which
$\Pr_z[\event_3] \geq \eta/(2L^2) = (d\cdot \maxdeg + \deg(\Qlcm))/q$,
we have at least $d\cdot \maxdeg + \deg(\Qlcm)$ many $z$'s satisfying
\eqref{eqn:composition-poly-r} and
thus the two polynomials in \eqref{eqn:composition-poly-r} are identical.
Let $G_{\alpha}(x)$ denote the the right-hand side of \eqref{eqn:composition-poly-r}
(replacing $z$ with $x$).

So far we have:
\[\Pr_{\alpha}[
    G_{\alpha}(x) \text{ is divisible by }
    \Qlcm(x)
] \geq \eta/(2L^2) > 1/q.\]

\lior{Check this paragraph}
Note that the set of $\alpha$'s satisfying this event forms a vector space.
If its dimension was less than $|\mathcal C|$ then the probability would have been
$\leq 1/q$. Hence this event holds for every $\alpha$.
Substituting the elements of the standard basis, we get that
for every $1 \leq i \leq |\mathcal C|$,
\[P^i(\tilde r^1(x \cdot M^i_1), \ldots,
\tilde r^1(x \cdot M^i_{|M^i|})) \cdot \frac{\Qlcm(x)}{Q^i(x)}
\text{ is divisible by } \Qlcm(x). \]
Substituting any $x$ for which $Q^i(x) = 0$ gives
$P^i(\tilde r^1(x \cdot M^i_1), \ldots, \tilde r^1(x \cdot M^i_{|M^i|})) = 0$
which implies that $\tilde r^1$ satisfies all the constraints, as required.
\end{proof}

\subsection{Further optimizations for practical implementation}

As we saw, it makes sense to
work with several witness polynomials rather than one, as it improves the prover complexity.
Another optimization is to apply the $\RPT$ only once
for both $h^1$ and $h^2$ by taking a random linear combination of the two
(and using \cref{lem:main}).
To make this work, the prover writes the degree $<d\cdot \maxdeg$
polynomial $\gpoly(x)$ as:
\[\gpoly(x) = \sum_{i=0}^{\maxdeg-1} x^i \gpoly_i(x^{\maxdeg}), \]
where the $\gpoly_i$s are of degree $<d$, and it sends oracles
to $\gpoly_i \mid_{D}$ instead of $\gpoly\mid_{D'}$.
In total, we will have to run $\RPT$ on $\width + \maxdeg$ polynomials
of degree $<d$, so we choose only one evaluation domain $D \subseteq \F_q$
satisfying $|D| = d\rho^{-1}$.

\bibliographystyle{alpha}
\bibliography{references}

\appendix

\section{Proof of the DEEP lemma for general codes}
\label{sec:DEEPgeneralproof}
\begin{proof}[Proof of \cref{lem:main general codes}]
	To simplify notation set $\eta = \max\left( 2L^*_\slack \left( \frac{\sigma}{N} + \epsilon\right)^{1/3}, \frac{4}{\epsilon^2 q} \right)$,
	and let $u_x = u^* + x u$.

	Let $\event[x,z]$ denote the event ``$\exists v \in \Ls(u_x,V,\slack), \inp{v,z}=B_z(x)$''.

	The assumption of Equation \eqref{eq1} now reads as
	$$\Pr_{x\in\F_q, z\in S}[\event[x,z]]\geq \eta.$$

	Thus we get,
	\begin{align}
	\label{generalzxmarkov}
	\Pr_{x\in \F_q} [ \Pr_{z\in S} [\event[x,z]]\geq \eta/2] \geq \eta/2
	\end{align}

	Let $$A = \condset{x\in \F_q}{\Pr_{z\in S} [\event[x,z]]\geq \eta/2]} $$
	and notice $|A|\geq \eta q/2$.

	For $x \in \F_q$, pick $v_x \in V$ to be a member of $\Ls(u_x,V,\slack)$ that maximizes
	$\Pr_{z \in S} [ P(z) = B_z(x)]$.
	Let $S_x = \condset{z \in S}{\inp{v_x,z} = B_z(x)}$ and set
	$\mu_x = |S_x|/s$.
	By definition, $|\Ls(u_x,V,\slack)| \leq L^*_\slack$, and so by the pigeonhole principle, for each $x \in A$ we have
	$\mu_x\geq \frac{\eta}{2L^*_\slack}$.

	For $x, \beta, \gamma$ picked uniformly from $A$ we have
	\begin{align*}
	\E_{x, \beta,\gamma\in A} \left[ \frac{|S_{x} \cap S_\beta \cap S_\gamma|}{s} \right] &= \E_{z\in S, x, \beta, \gamma\in \F_q} [ 1_{z \in S_x \cap S_\beta \cap S_{\gamma}} ]\\
	&= \E_{z\in S} [ \E_{x\in \F_q} [ 1_{z \in S_{x}}]^3 ]\\
	&\geq \E_{z\in S, x\in \F_q} [ 1_{z \in S_{x}} ]^3 \\
	&\geq \left(\frac{\eta}{2L^*_\slack}\right)^3\\
	&> \frac{\sigma}{N} + \epsilon.
	\end{align*}


	The second equality above follows from the independence of
	$x, \beta, \gamma$. The first inequality is an application of Jensen's inequality and the last inequality is by assumption on $\eta$.

	Thus $$\Pr_{x, \beta, \gamma} [ |S_{x} \cap S_\beta \cap S_\gamma| > \sigma ] \geq \epsilon .$$

	Note that $\Pr_{x,\beta,\gamma}[x,\beta,\gamma \text{ are not all distinct}] < 3/|A|$.
	Since $|A| \geq \eta q/2 \geq 2/\epsilon^2 \geq 6/\epsilon$ we have $3/|A| \leq \epsilon/2$.
	Thus $\Pr_{x,\beta,\gamma}[x,\beta,\gamma \text{ are all distinct and }|S_x \cap S_\beta \cap S_\gamma| > \sigma] \geq \epsilon/2$.

	This means that there are distinct $x_0, \beta_0$
	such that
	$$ \Pr_{\gamma} [  |S_{x_0} \cap S_{\beta_0} \cap S_\gamma| > d] \geq \epsilon/2.$$

	Consider some $\gamma$ where this happens. Let $\tilde{S} = S_{x_0} \cap S_{\beta_0} \cap S_{\gamma}$.
	Extend each of $u^*, u$ to functions over domain $S$ by defining for all $z\in S\setminus [n]$ $u^*(z)=B_z(0)$ and $u(z)=B_z(1)$, and for $x\in \F_q$ let $u_x(z)=u^*(z)+x u(z)$. Since $V$ is systematic, we define $v_x(z)=\inp{v_x|_{[k]},z}$ and thus extend $v_x$ to domain $\tilde{S}$.
	By construction we know
	$$(x_0, u_{x_0}), (\beta_0, u_{\beta_0}), (\gamma, u_\gamma)$$
	are collinear. So, in particular,
	$$(x_0, u_{x_0}|_{\tilde{S}}), (\beta_0, u_{\beta_0}|_{\tilde{S}}), (\gamma, u_{\gamma}|_{\tilde{S}}) \in \F_q \times \F_q^{\tilde{S}}$$
	are likewise collinear, as a special case.
	By definition of ${\tilde{S}}$, we get that:
	$$(x_0, v_{x_0}|_{\tilde{S}}), (\beta_0, v_{\beta_0}|_{\tilde{S}}), (\gamma, v_{\gamma}|_{\tilde{S}}) \in \F_q \times \F_q^{\tilde{S}}$$
	are also collinear.
	Since $|{\tilde{S}}| > \sigma$ and $S$ is $\sigma$-robust we conclude that $v_\gamma$ is uniquely determined by $v_\gamma|_{\tilde{S}}$. This allows us to conclude that
	$$(x_0, v_{x_0}), (\beta_0, v_{\beta_0}), (\gamma, v_{\gamma}) \in \F_q \times \F_q^n$$
	are all collinear, recalling that $v_{x_0}\in \Ls(u_{x_0},V,\slack)$.

	Thus, an $\epsilon/2$-fraction of the $\gamma \in A$ have the ``good'' property that
	$(\gamma, v_{\gamma})$ is on the line passing through
	$(x_0, v_{x_0})$ and $(\beta_0, v_{\beta_0})$.
	Write this line as $v^* + x v$ and notice that for all ``good'' $\gamma$ we have $v_\gamma=v^*+\gamma v$.
	Let $A'\subseteq A$ denote the set of good elements for this line, recording that $|A'|\geq |A|\cdot \epsilon/2\geq 1/\epsilon$.
	By definition of $\Ls(u_x,V,\slack)$ and
	the assumption $v_x\in \Ls(u_x,V,\slack)$,
	we have that $\Delta(u_x, v_x) < \slack$ for $x\in A'$.

	Consider the set $C\subset [n]$ defined by
	$$C = \condset{y\in [n]}{u^*(y)=v^*(y) \mbox{ AND } u(y)=v(y)}.$$

	For each $y \in [n]\setminus C$ there exists at most a single value of $x\in \F_q$ satisfying $u_x(y)=v_x(y)$ because
	$$u_x(y)-v_x(y)=(u^*(y)-v^*(y)) + x\cdot(u(y)-v(y))$$
	has at most one value $x$ on which it vanishes. This implies
	$$\slack\geq \E_{x\in A'}[\Delta_{[n]}(u_x,v_x)]\geq
	\frac{|[n]\setminus C|}{n}\cdot \left(1-\frac{1}{|A'|}\right)\geq \left(1-\frac{|C|}{n}\right)\cdot(1-\epsilon)\geq 1-\frac{|C|}{n}-\epsilon.$$
	Rearranging, we get $\frac{|C|}{n}\geq 1-(\slack+\epsilon)$ and this completes the proof.
\end{proof}

\section{The algebraic hash function}
\label{sec:alghash}
We now describe the algebraic hash function $H_x$.

The description of the hash function requires
fixing some choices of certain subspaces. For each $i \in [0, \rounds]$
we choose $\F_2$-subspaces $L\ii_0$ and $L\ii$,
satisfying the following properties.
\begin{enumerate}
	\item $L\ii_0 \subseteq L\ii$ with $\dim(L\ii_0) = 1$,
	\item $L\iip = q\ii(L\ii)$, where $q\ii(X)$ is the
	{\em subspace polynomial} of $L\ii_0$,
	$$q\ii(X) = \prod_{\alpha \in L\ii_0} (X-\alpha),$$
	thus this is an $\F_2$-linear map with kernel $L\ii_0$).
	In particular, $\dim(L\iip) = \dim(L\ii) - 1$.
\end{enumerate}
Let $\cosets\ii$ denote the set of cosets of $L\ii_0$ contained in $L\ii$.

Given $x \in \F$ and $f:L\ii \to \F$, the
hash of $f$ with seed $x$ is defined to be the function $H_{x}[f]:L\iip\to\F$ as follows.
For $s \in L\iip$, let $s_0,s_1 \in L\ii$ be the two roots of $q\ii(X)-s$.
Let $P_{f, s}(X) \in \F[X]$ be the unique degree $\leq 1$ polynomial satisfying
$$P_{f,s}(s_0) = f(s_0),$$
$$P_{f,s}(s_1) = f(s_1).$$
Then we define
\begin{equation}
H_x[f](s) = P_{f,s}(x).
\end{equation}
Observe that $H_x[f](s)$ can be computed by
querying $f$ on the set $\set{s_0,s_1}$ (this set
is a coset of $L\ii_0$, and we denote it by $\coset\ii_{s}$).

To understand $H_x$ better, it is instructive to see what it does to
$\RS\ii$.
Let $f \in \RS\ii$. The underlying polynomial $f(X)$ thus has degree at most $\Rate |L\ii|$.
We may write $f(X)$ in base $q\ii(X)$ as:
\begin{equation}
\label{fXqX}
f(X) = a_0(X) + a_1(X) q\ii(X) + \ldots + a_t(X) (q\ii(X))^t,
\end{equation}
where each $a_i(X)$ has degree at most $1$, and
$t \leq \Rate |L\ii|/2$.
Since the polynomials $f(X)$ and $P_{f,s}(X)$ agree on the roots of
$q(X) - s$, we get that $f(X) \equiv P_{f,s}(X) \mod (q\ii(X)-s)$.
From Equation~\eqref{fXqX}, we get that
$$P_{f,s}(X) = a_0(X) + a_1(X) s + \ldots + a_t(X) s^t.$$
In particular, for all $x \in \F$,
$$H_x[f](s) = P_{f,s}(x) = a_0(x) + a_1(x) s + \ldots + a_t(x) s^t,$$
and thus
$$H_x[f] \in \RS\iip.$$

\section{Proof of \cref{alphabeta1} and \cref{alphabeta2}}
\label{sec:alphabetalemmas}

We first prove~\cref{alphabeta1}.
\begin{proof}
Set $\gamma = \max(\alpha\ii, 1-\deltastar )$.

For simplicity, denote $f\ii$ by $f$.

 Recall the notation $P_{f,s}$ from the definition of the algebraic hash function $H_x$ in Section~\ref{sec:alghash}. We have that for each $s \in L\iip$, $H_x[f](s) = P_{f,s}(x)$ is a linear function of $x$.
Thus we can write $H_x[f] = u^* + x u$ for $u^*, u \in \F_q^{L\iip}$, and for any
fixed $s$, we have the formal polynomial equality $P_{f,s}(X) = u^*(s) + X u(s)$.

We are interested in bounding the probability of the event $\beta\iip > \gamma + \epsilon$.
In other words, we want to bound the probability that there exists a polynomial $Q(Y) \in \F_q[Y]$
with $\deg(Q) < d\iip+1$ such that:
\begin{itemize}
\item $\agree_{\theta\iip}(u^* + xu, Q) > \gamma + \epsilon$,
\item $Q(z\ii) = B\ii_{z\ii}(x)$.
\end{itemize}

This is exactly the scenario of~\cref{lem:weighted-OOD-main}.
That Lemma tells us that if the probability in question is larger than
$\nustar$, then there exist polynomials $P(Y), P^*(Y)$ of degree $\leq d\iip $ and a set $T \subseteq L\iip$
such that:
\begin{itemize}
\item $$\frac{1}{|L\iip|}\sum_{s \in L\iip}\theta\iip > \gamma,$$
\item $u|_T = P|_T$,
\item $u^*|_T = P^*|_T$.
\end{itemize}

    Let
    $$\hat{P}(X,Y)\eqndef P^*(Y)+X\cdot P(Y)$$
    and notice that $\deg_X(\hat{P})\leq 1$, $\deg_Y(\hat{P}) \leq d\iip$.

    Consider the polynomial $R(X)\eqndef \hat{P}(X,q\ii(X))$. We have $$\deg(R)\leq 2 d\iip + 1 = d\ii -1 < d\ii.$$

    We claim that $R$ agrees with $f$ on $\tilde{T} = \bigcup_{s \in T}\coset\ii_{s}$.

Take any $s \in T$ and let $\coset\ii_{s}=\set{s_0, s_1}\in \cosets\ii $ be the pair of roots of the polynomial
$q\ii\left(X\right)-s$.

First we show that the polynomials $P_{f,s}(X)$ and $\hat{P}(X,s)$ are identical.
Indeed, $\hat{P}(X, s)) =  P^*(s) + X P(s) = u^*(s) + X u(s) =  P_{f,s}(X)$.
  It follows that
		$$f\left(s_0\right)=
    \hat{P}\left(s_0,s\right)=
    \hat{P}\left(s_0,q\ii\left(s_0\right)\right)=
    R\left(s_0\right)$$ and similarly $f\left(s_1\right)=R\left(s_1\right)$. Therefore, $R$ and $f$
    agree on $\tilde{T}$, as claimed.

    We now use the above information to show that $\alpha\ii = \agree_{\eta\ii}(f, R) > \gamma$, which contradicts the definition of $\gamma$.
Indeed,
\begin{align*}
\agree_{\eta\ii}(f,R) &= \frac{1}{|L\ii|} \sum_{r \in L\ii \mid f(r) = R(r)} \eta\ii(r) \\
&\geq \frac{1}{|L\ii|} \sum_{r \in \tilde{T}} \eta\ii(r) \\
&= \frac{1}{|L\ii|} \sum_{s \in T} \sum_{r \in \coset\ii_s} \eta\ii(r) \\
&= \frac{1}{|L\ii|} \sum_{s \in T} |\coset\ii_s| \cdot \theta\ii(s)  \quad\mbox{Since $\theta(s)$ equals the average of $\eta(r) \mid r \in \coset\ii_s$}\\
&= \frac{1}{|L\iip|} \sum_{s \in T} \theta\ii(s)  \\
&> \gamma.
\end{align*}
This is the desired contradiction.
\end{proof}

Next we prove~\cref{alphabeta2}.
\begin{proof}
By definition,
$$ \beta\ii = \agree_{\theta\ii}(H_{x\iim}[f\iim],  \{P(Y) \in \F_q[Y] \mid \deg(P) \leq d\ii \mbox{ and } P(z\iim) = B\iim_{z\iim}(x\iim)\} )$$

Next, by the properties of quotienting, \cref{lem:quotient},
\begin{align*}
\beta\ii &= \agree_{\theta\ii}( H_{x\iim}[f\iim],  \{P(Y) \in \F_q[Y] \mid \deg(P) \leq d\ii \mbox{ and } P(z\iim) = B\iim_{z\iim}(x\iim)\})\\
&= \agree_{\theta\ii}( \QUOTIENT(H_{x\iim}[f\iim], z\iim, B\iim_{z\iim}(x\iim)),  \{P(Y) \in \F_q[Y] \mid \deg(P) \leq d\ii-1 \}).
\end{align*}

Now observe that $\eta\ii$ is obtained from $\theta\ii$ by zeroing out coordinates in $E\ii$, and
the only coordinates where $f\ii$ can differ from $\QUOTIENT(H_{x\iim}[f\iim], z\iim, B\iim_{z\iim}(x\iim))$
are in $E\ii$.
Thus:
\begin{align*}
\beta\ii &\geq  \agree_{\theta\ii}( f\ii,  \{P(Y) \in \F_q[Y] \mid \deg(P) \leq d\ii-1 \})\\
&= \agree_{\theta\ii}( f\ii, \RS\ii)\\
&= \alpha\ii.
\end{align*}
This completes the proof.
\end{proof}

\end{document}

\end{document}

\newpage

\section{More (heuristic) counterexamples to high-error distance preservation}

\swastik{TO BE CLEANED UP OR OMITTED FOR THIS VERSION -- most likely omitted, but let us bring it back for arxiv}

Do the same with subspace polynomials of degree $2^{n-c}$ (with $c$ constant).
There are $\approx 2^{c(n-c)}$ such subspace polynomials.
They all look like:
$$ Y^{2^{n-c}} + \sum_{j=1}^{c-1} b_j Y^{2^{n-c-j}} + \beta Y^{2^{n-2c}}
+ (\mbox{degree $< 2^{n-2c-1}$}).$$
As the subspace varies, we expect that the the $b_j$ and the $\beta$ are all
pretty random. Thus we expect that there is some choice of $(b_1, \ldots, b_{c-1})$ such that there are $\Omega(2^n)$ choices of $\beta$ for which
$(b_1, \ldots, b_{c-1}, \beta)$ are the leading coefficients of a subspace polynomial.

Using this, we get the following.
For $\rho = 2^{-2c-1}$, there are functions $u^*$ and $u$ (given by $u^*(y) = y^{2^{n-2c}}$ and $u(y) = y^{2^{n-c}} + \sum_{j=1}^{c-1} b_j Y^{2^{n-c-j}}$) such that:
$$\Delta(u^*, \RS(\rho)) \geq 1- 2^{-2c},$$
but:
$$\Pr_{x \in \F_q} [ \Delta(u + x u^*, \RS(\rho)) \leq 1- 2^{-c} ] \geq \Omega(1).$$

{\bf THIS IS NOT FULLY PROVED - IT DEPENDS ON THE ``RANDOMNESS" of the leading $c$ coefficients of subspace polynomials of subspaces of dimension $n-c$}. This is quite believable, and we can probably code up examples that work based on this.

For $c=1$ we get the original example.

For $c=2$, we get the following: when $\rho = 1/32$, there is a function $u^*$ with agreement $1/16$ with the corresponding $\RS$ code, but there is a line through $u^*$ where constant fraction of the points on the line have agreement $1/4$ with $\RS$.

In general, for various $\rho$, we have a $u^*$ with agreement $2\rho$, but many points on the line have agreement $\sqrt{2\rho}$.

\end{document}